\documentclass[journal]{IEEEtran}

\usepackage[final]{graphicx}
\graphicspath{{Images//}} 

\usepackage{psfrag}
\usepackage{amsfonts}
\usepackage{mathrsfs}
\usepackage{pifont}
\usepackage{verbatim}
\usepackage{upgreek}
\usepackage{color,soul}
\soulregister\cite7
\soulregister\ref7
\soulregister\eqref7
\usepackage[usenames,dvipsnames]{xcolor}
\usepackage{caption}
\usepackage{algorithmic}
\usepackage{algorithm}
\usepackage{mdwmath}
\usepackage{mdwtab}
\usepackage{amssymb,amsmath}
\usepackage{amsmath}
\usepackage{amsthm}
\usepackage{epstopdf}
\usepackage{cite}
\usepackage{tikz}
\usepackage{scalefnt}
\usepackage{subfigure}
\usepackage{MnSymbol}
\usepackage{setspace}

\label{newcommand}
\label{English Chars}
\newcommand{\bA}{\mbox{\boldmath{$A$}}}
\newcommand{\ba}{\mbox{\boldmath{$a$}}}

\newcommand{\bb}{\mbox{\boldmath{$b$}}}

\newcommand{\bR}{\mbox{\boldmath{$R$}}}

\newcommand{\bs}{\mbox{\boldmath{$s$}}}

\newcommand{\bW}{\mbox{\boldmath{$W$}}}
\newcommand{\bw}{\mbox{\boldmath{$w$}}}

\newcommand{\bx}{\mbox{\boldmath{$x$}}}

\newcommand{\by}{\mbox{\boldmath{$y$}}}

\newcommand{\bz}{\mbox{\boldmath{$z$}}}
\newcommand{\diag}{\mbox{ diag}\, }

\label{Roman Chars}
\newcommand{\balpha}{\mbox{\boldmath{$\alpha$}}}
\newcommand{\bbeta}{\mbox{\boldmath{$\beta$}}}

\newcommand{\btheta}{\mbox{\boldmath{$\theta$}}}

\newcommand{\bPhi}{\mbox{\boldmath{$\Phi$}}}
\newcommand{\bPsi}{\mbox{\boldmath{$\Psi$}}}

\newcommand{\bLambda}{\mbox{\boldmath{$\Lambda$}}}

\newcommand{\bepsilon}{\mbox{\boldmath{$\epsilon$}}}

\newcommand{\bpsi}{\boldsymbol \psi}
\newcommand{\bTheta}{\mbox{\boldmath{$\Theta$}}}

\newcommand{\bxi}{\mbox{\boldmath{$\xi$}}}

\newcommand{\blambda}{\mbox{\boldmath{$\lambda$}}}
\newcommand{\bnu}{\mbox{\boldmath{$\nu$}}}

\label{Theorems}
\newtheorem{theorem}{Theorem}[]

\newtheorem{remark}{Remark}
\begin{document}
% Title
\label{title}
%\title{Clutter suppression using Capon Beamforming in Colocated MIMO radar with Compressed Domain Processing}
\title{Compressed-Domain Detection and Estimation for Colocated MIMO Radar}
\author{Ehsan~Tohidi\IEEEauthorrefmark{1}, Alireza Hariri\IEEEauthorrefmark{2}, Hamid~Behroozi\IEEEauthorrefmark{2},
        Mohammad~Mahdi~Nayebi\IEEEauthorrefmark{2},\\  Geert~Leus\IEEEauthorrefmark{3}, Athina Petropulu\IEEEauthorrefmark{4},%, and Athina Petropulu\IEEEauthorrefmark{2}%~\IEEEmembership{Senior Member,~IEEE}
\thanks{\IEEEauthorrefmark{1}Department of Communication Systems, Eurecom, Biot, France. E-mail: tohidi@eurecom.fr.}
\thanks{\IEEEauthorrefmark{2}Department of Electrical Engineering, Sharif University of Technology, Tehran, Iran.}
\thanks{\IEEEauthorrefmark{2}Faculty of Electrical Engineering, Mathematics and Computer Science, Delft University of Technology, Delft, The Netherlands.}
\thanks{\IEEEauthorrefmark{3}Department of Electrical and Computer Engineering,
	Rutgers—The State University of New Jersey, New Brunswick, USA.}
}
\maketitle

\begin{abstract}
\boldmath 			
	%While multiple input multiple output (MIMO) radar achieves high angle resolution target detection with a small number of antennas, they involve high sample and computational complexity.
	%By exploiting the sparsity of targets in the target space, compressed sensing (CS) techniques have been used in the context of MIMO radar to significantly reduce the amount of samples required to achieve the resolution of MIMO radar. However,  existing CS MIMO radar techniques use the compressively obtained samples to recover the Nyquist rate target scene, where they perform target detection. 
	%Thus, while fewer samples need to be acquired, target detection still involves high sample and computation complexity. 
	{This paper proposes compressed domain signal processing (CSP) multiple input multiple output (MIMO) radar, a MIMO radar approach that achieves substantial sample complexity reduction by exploiting the idea of CSP. CSP MIMO radar involves two levels of data compression followed by target detection at the compressed domain. First, compressive sensing is applied at the receive antennas, followed by
	a Capon beamformer which is designed to suppress clutter. Exploiting the sparse nature of the beamformer output, a second compression is applied to the filtered data. Target detection is subsequently conducted by formulating and solving a hypothesis testing problem at each grid point of the discretized angle space. The proposed approach enables an 8-fold reduction of the sample complexity in some settings as compared to a conventional compressed sensing (CS) MIMO radar thus enabling faster target detection. Receiver operating characteristic (ROC) curves of the proposed detector are provided. Simulation results show that the proposed approach outperforms recovery-based compressed sensing algorithms.}
\end{abstract}

\section{Introduction}	\label{Introduction}
\IEEEPARstart{T}{he} emergence of multiple input multiple output (MIMO) radar opened up a wide research area, promising the same resolution as phased array technology but with significantly fewer antennas elements, or higher resolution with the same number of antennas. MIMO radar transmit different waveforms from their antennas. Based on antennas distances, MIMO radar is categorized into {\it widely separated} and {\it colocated}. Large distances among antennas in widely separated MIMO radar cause different transmitter-receiver pairs to look at a target from different angles; this provides spatial diversity and results in high-resolution target localization and enhanced target detection and estimation \cite{haimovich2008mimo,Radieee2,tohidi2017compressive}. In colocated MIMO radar, exploiting waveform diversity results in flexible beampattern design and improved angular resolution \cite{li2007mimo, karbasi2015design,8335404,8706630}. In this paper, we focus on colocated MIMO radar.
Despite the many advantages, the requirement for a large amount of data and associated computational complexity are viewed as the main drawbacks of MIMO radar \cite{8537943,rossi2014spatial}. Fortunately, due to the low number of targets in the target space (angle, range, speed), the target echoes are sparse \cite{ender2010compressive, tohidicosera,biondi2015compressed,7745736,8514046,8639010}. This characteristic enables the incorporation of compressed sensing (CS) theory, which, under certain conditions, allows for lower than Nyquist rate sampling with a negligible performance reduction \cite{baraniuk2007compressive,emmanuel2004robust}. A general discussion of CS applied to radar can be found in \cite{ender2010compressive,yu2010mimo}.
MIMO radar's ability to achieve high angle resolution with
small numbers of elements renders them indispensable for
automotive applications. This advantage has been exploited by
almost all major automotive Tier-1 suppliers in their different
types of radar products, such as SRR, MRR and LRR \cite{7485215,alland2009radar,wintermantel2014radar,schoor2016method}. %\hlblue{It is worth mentioning that this reduction in the sampling rate is not for free and fundamentally, CS often causes a loss in resolution in comparison with a non-CS setting. We illustrate this CS loss in the framework of the proposed problem.}

CS application to MIMO radar has received a lot of attention recently, e.g., \cite{8361480,7467561,gogineni2011target,7376237,herman2009high,rossi2014spatial,yu2014power,lei2016compressed,8468214,8395364}. For instance, target detection and localization in MIMO radar using CS is discussed in \cite{gogineni2011target,7376237}, while improving angular resolution with a lower number of elements in a colocated MIMO radar is studied in \cite{herman2009high,rossi2014spatial,8361480}. Similarly, power allocation and waveform design in CS MIMO radar is investigated in \cite{yu2014power,7467561}.
In all aforementioned works, the signal used for detection and/or estimation is first reconstructed by using a general-purpose CS recovery algorithm such as orthogonal matching pursuit (OMP) \cite{gogineni2011target}, basis pursuit de-noising or compressive sampling matching pursuit \cite{needell2009cosamp}, alternating direction method of multipliers (ADMM) \cite{7376237}, or problem-specific algorithms \cite{huang2012adaptive,zhang2011sparse}. 
In the cited methods, CS is used to reduce the amount of data collected and transmitted to a fusion center, where sparse signal recovery is carried out. However, recovering the sparse signal switches the problem back to the high rate domain, thus does not take full advantage of the CS enabled reduction of large amounts of data. In many radar applications, the original signal may not be of interest, and the main aim is to accomplish radar inference tasks (e.g., detection and estimation). Therefore, signal processing in the CS domain (i.e., without reconstruction) is desirable. Note that from an information theoretic aspect, signal reconstruction does not increase the available information. Further, if the sensing matrix does not have low coherence, the recovery may be incorrect.
In this paper,  we go one step further on the use of the CS idea; we do not recover the underlying sparse signal, but rather perform target detection in the compressed domain, using compressed domain signal processing (CSP). We show that CSP based target detection and parameter estimation not only preserves the performance and significantly reduces the number of computations, but also prevents the high flow of data after recovery which is one of the fundamental motivations of employing CS \cite{davenport2010signal}.

CSP has been studied in various applications. For instance, CSP is used in \cite{hariri2017compressive} to detect sparse signals in additive white Gaussian noise and estimate the degree of sparsity. Similarly, a CSP based symbol detector for UWB communications is proposed in \cite{6797969}.
Also, CSP is used in \cite{hariri2015joint} to accomplish joint compressive single target detection and parameter estimation in a radar. Algorithms for solving inference problems such as detection, classification, estimation, and filtering based on CSP are proposed in \cite{hariri2015joint}, while in \cite{wicks2012compressed}, the idea of using CSP for space-time adaptive processing is presented. The task of inferring the modulation of a communication signal directly in the compressed domain is considered in \cite{lim2012automatic,lim2012chocs}. Furthermore, in \cite{5604239}, a minimum variance distortionless response (MVDR) beamformer is used in the compressed domain for the task of spectrum sensing. However, clutter is not considered in the signal model and the performance is only evaluated through simulations.

Clutter is a critical nuisance component in radar signal processing \cite{skolnik2001introduction}, and clutter suppression is a very important task \cite{4358826,4768708,6589175,7470499}. 
Clutter changes the target scene, making it less sparse. Therefore, the performance of CS-based radar detection methods deteriorates in the presence of clutter. 
The Capon beamformer, also known as the MVDR is a common clutter suppression approach that relies on the availability of clutter statistics (i.e., the clutter covariance matrix). 
In the context of CS-based colocated MIMO radar, \cite{yu2013capon} applies Capon beamforming on the compressed clutter contaminated target echoes, before proceeding with CS-based sparse signal recovery.
In this paper, we consider the same scenario as in \cite{yu2013capon}. We apply Capon beamforming on the compressed radar returns, but unlike \cite{yu2013capon}, we proceed with target detection by operating directly in the compressed domain. 

\subsection{Contributions}
The main contribution of this paper is a CSP approach for detection and parameter estimation of a noise and clutter contaminated target in a colocated MIMO radar scenario.
In particular, 
\begin{itemize}
	\item We formulate and solve a hypothesis testing problem by operating in the compressed samples domain.
	\item We employ a Capon beamformer as a preprocessing step to reduce the clutter power. The beamformer sparsifies the target scene, which allows us to use a second compression at the beamformer output, thus achieving further sample complexity savings.
	\item Through receiver operating characteristic (ROC) analysis, and also simulations, we illustrate that CSP MIMO radar performs well achieving a 8-fold sample complexity reduction in some settings as compared to recovery-based methods.
		This translates to faster detection, making the proposed approach a good candidate for low latency applications, such as automotive radar.
		Interestingly, in addition to having lower complexity, the proposed approach outperforms recovery-based algorithms in terms of angle estimation accuracy in the case of multiple targets.
\end{itemize}

\subsection{Outline and Notations}
The rest of the paper is organized as follows. Section II provides the required preliminaries. The signal model is introduced in Section III. The CSP algorithm is proposed in Section IV. Simulation results are reported in Section V. Section VI concludes the paper.

We adopt the notation of using boldface lower case for vectors $\ba$, and bold face upper case for matrices $\bA$, where $\ba_i$ is the $i$th column of the matrix $\bA$. The transpose, Hermitian, complex conjugate, and pseudo inverse operators are denoted by the symbols $(.)^T$, $(.)^H$, $(.)^*$, and $(.)^\dagger$, respectively. Given a set of indices $\mathcal{S}$, $\bA[\mathcal{S}]$ is a matrix composed of the columns of $\bA$ with indices in the set $\mathcal{S}$. $\mathbb{R}^{N \times M}$ and $\mathbb{C}^{N \times M}$ are the set of $N\times M$ real and complex matrices, respectively. Finally, $\diag(\bA_1,...,\bA_N)$ indicates the block diagonal matrix formed by the matrices $\bA_1,...,\bA_N$ along the main diagonal.

\section{System Model}
Consider a colocated MIMO radar with $I$ transmitters and $R$ receivers. We assume that the transmitters and receivers form a uniform linear array with $\lambda/2$ spacing, where $\lambda$ is the wavelength. The antenna configuration is shown in Figure \ref{fig:config} (similar to the configuration in \cite{6650099,7272834}). Let $s_i(n)$ denote the discrete-time baseband signal transmitted by the $i$th transmitter. The transmit steering vector is given by

\begin{equation}
\ba(\theta) = [\exp(j2\pi f_c \tau_1(\theta)), ..., \exp(j2\pi f_c \tau_{I}(\theta))]^T,
\end{equation}
where $\tau_i(\theta), i=1,...,I$, is the propagation delay from the $i$th transmitter to the target at angle $\theta$ with respect to the array axis and $f_c$ is the carrier frequency (i.e., $f_c=c/\lambda$, where $c$ is the light speed). Furthermore, the receive steering vector is

\begin{equation}
\bb(\theta) = [\exp(j2\pi f_c \tilde{\tau}_1(\theta)), ..., \exp(j2\pi f_c \tilde{\tau}_{R}(\theta))]^T,
\end{equation}
where $\tilde{\tau}_r(\theta), r=1,...,R$, is the propagation delay from a target at angle $\theta$ to the $r$th receiver. Let us assume that there are $Q$ targets in the region of interest. On sampling the received signals with Nyquist sampling interval $T_s$, the obtained samples across all receivers at sampling instance $n$, i.e., $x_1(n), ..., x_R(n)$, can be expressed in a vector form as
\begin {equation}
\bx(n) = \sum_{q=1}^Q{\alpha_q \bb(\theta_q)\ba^H(\theta_q)\bs(n) + \bepsilon(n)}, n=1,...,N,
\label{equ:ktarget}
\end {equation}
where $\bx(n) = [x_1(n), ..., x_{R}(n)]^T$; {$\alpha_q,~q=1,...,Q$, is the complex amplitude of the $q$th target as seen by the receivers (due to the colocated MIMO radar assumption, the radar cross section (RCS) of each target seen by all transmitter-receiver pairs is the same).} Here it is assumed that the $\alpha_q$'s are constant during the observation interval (i.e., Swerling I model); $\bs(n)=[s_1(n),...,s_I(n)]^T$ is the $n$th time sample of the transmit signal vector; and $\bepsilon(n)$ is the noise plus clutter term at the receivers. $\bepsilon(n)$ is assumed to be complex Gaussian with covariance matrix $\bR_N\in \mathbb{C}^{R\times R}$, i.e., $\bepsilon(n) \sim \mathcal{CN}(0,\bR_N)$.
Also, we assume that $\alpha_q \sim \mathcal{CN}(0,\sigma_{\alpha}^2)$ \cite{skolnik2001introduction}.
For simplicity, let us assume that target's angle is the only parameter of interest. Adding velocity to \eqref{equ:ktarget} would be addressed in a similar fashion, i.e., it would amount to adding velocity to the hypothesis test and searching in angle-velocity space for both detection and estimation tasks. Also, similar to other works on the target angle of arrival estimation (e.g., \cite{rossi2014spatial}), the data vector in \eqref{equ:ktarget} is considered for a specific range cell, and therefore the delay is known and can be compensated (Figure \ref{fig:system_model}). Henceforth, to cover the whole range space, the entire procedure of detection and estimation would have to be performed separately for each range cell.

Suppose the angle space of interest has been discretized into $L$ uniform grid angles $\mathcal{L}=\{\theta_1,...,\theta_L\}$ and the targets lie on the grid. Then, (\ref{equ:ktarget}) can be reformulated as

\begin{equation}
\bx(n) = \bPsi(n) \bbeta + \bepsilon(n),
\end{equation}
where $\bPsi(n)\in \mathbb{C}^{R\times L}$ is the $n$th sample of the measurement matrix in which the $l$th column, $\bpsi_l(n)$, is parametrized based on the grid angle $\theta_l$, and equals
\begin{equation}
\bpsi_l(n) = \bb(\theta_l)\ba(\theta_l)^H\bs(n).
\end{equation}
Moreover, $\bbeta\in\mathbb{C}^{L\times 1}$ is the target amplitude vector, determined as
\begin {equation}
\begin{aligned}
 \beta_l &=
  \begin{cases}
   \alpha_q  & \text{if the $q$th target is at angle } \theta_l \\
   0        & \text{otherwise}.
  \end{cases}
\end {aligned}
\end {equation}
Stacking the signals of the $N$ Nyquist samples obtained by all antennas, the total received data vector is given by
\begin{equation}
\label{regularmeasures}
\bx
= \begin{bmatrix}
{\bPsi(1)} \\
\vdots \\
{\bPsi(N)}
\end{bmatrix}\bbeta
+ \begin{bmatrix}
{\bepsilon}(1) \\
\vdots \\
{\bepsilon}(N)
\end{bmatrix}
= {\bPsi}\bbeta + {\bepsilon},
\end{equation}
where ${\bPsi}\in\mathbb{C}^{RN\times L}$ is the total measurement matrix and
$\bepsilon$ is a complex Gaussian vector with a covariance matrix $\tilde{\bR}_N\in \mathbb{C}^{RN\times RN}$, i.e., $\bepsilon \sim \mathcal{CN}(0,\tilde{\bR}_N)$.

If there is a small number of targets within the range cell under investigation, $\beta$ will be sparse \cite{yu2010mimo}. This implies that under certain conditions \cite{yu2010mimo}, all information about $\beta$ is retained in the compressed vector $\bar{\bx}\in \mathbb{C}^{M_1\times 1}$
for which it holds that
\begin{equation}
\bar{\bx} = \bPhi_{(1)} {\bPsi} \bbeta + \bPhi_{(1)} {{\bepsilon}} = \bLambda \bbeta + {\bxi},
\label{equ:cmvector}
\end{equation}
where $\bPhi_{(1)} \in \mathbb{R}^{M_1\times RN}$ with $M_1 < RN$ is the compression matrix performing a joint temporal and spatial CS, along the time and array domains, respectively, and we define the first compression ratio $\text{CR}_1= \frac{RN}{M_1}$ as the ratio of the number of samples in regular sensing, $RN$, to the number of compressed measurements, $M_1$.
In addition, $\bLambda=\bPhi_{(1)} {\bPsi}\in\mathbb{C}^{M_1\times L}$ and ${\boldsymbol \xi} = {\boldsymbol \Phi}_{(1)} {\boldsymbol \epsilon} \in {\mathbb C}^{M_1 \times 1}$. 

In the following, we address the problem of detecting whether a target exists within the grid angles and if it does, estimating the target's angle by operating at the compressed samples domain.

%\begin{figure}
%	\centering
%	\includegraphics[width=.44\textwidth]{configuration}
%	\caption{MIMO radar system model.}
%	\label{fig:config}
%\end{figure}

\begin{figure}
	\centering
	\psfrag{Target}{\scriptsize{Target}}
	\psfrag{Transmitters}{\scriptsize{Transmitters}}
	\psfrag{Receivers}{\scriptsize{Receivers}}
	\psfrag{Range cell under test}{\scriptsize{Range cell under test}}
	\psfrag{T1}{\scriptsize{$\theta_1$}}
	\psfrag{T2}{\scriptsize{$\theta_2$}}
	\psfrag{T3}{\scriptsize{$\theta_3$}}
	\psfrag{TL}{\scriptsize{$\theta_L$}}
	\psfrag{dots}{\scriptsize{$\vdots$}}
	\subfigure[] {\includegraphics[width=.44\textwidth]{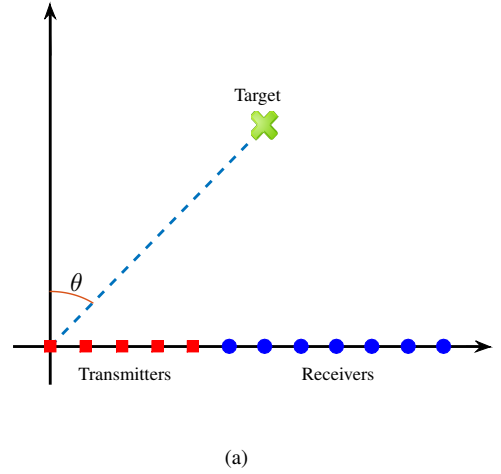}
		\label{fig:config}} \quad
	\subfigure[] {\includegraphics[width=.51\textwidth]{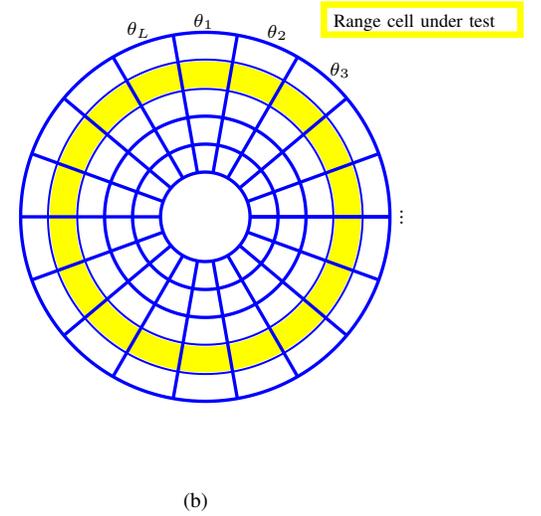}
		\label{fig:system_model}} \\
	\caption{MIMO radar system model, (a) configuration of the radar with $I$ transmitters and $R$ receivers placed in a uniform linear array with $\lambda/2$ spacing, (b) all the range and angle cells where a specific range cell (i.e., the colored one) is under test.}
	\label{fig:configuration}
\end{figure}

\section{The proposed approach}
In a practical setting, the received data are contaminated by clutter, which destroys the sparsity of the measured signal $\bar{\bx}$. Typically, the clutter arises due to reflections by the landscape and and thus can be studied when no targets exist. Here, we assume that statistical information about the clutter is available in the from of a clutter covariance matrix. In such case, the Capon's MVDR beamformer \cite{van2002optimum} can be constructed and applied to the obtained data to reduce clutter and thus make the scene sparser.
The objective of Capon's MVDR beamformer is to design a filter so that at the filter output, the noise and clutter power is minimized, while leaving the desired signal without distortion \cite{van2002optimum}.
The Capon weighting vector for each angle cell is obtained by solving the following optimization problem

\begin{equation}
\begin{aligned}
& \underset{\bw_l}{\min}
& & \bw_l^H\bR_C\bw_l \\
& \text{subject to}
& & \bw_l^H\blambda_l = 1,
\end{aligned}
\label{caponoptim}
\end{equation}
where $\bw_l$ is the weighting vector matched to the $l$th angle cell and ${\bR}_C = \bPhi_{(1)}\tilde{\bR}_N\bPhi_{(1)}^T$ is the covariance matrix of the measured clutter and noise, $\bxi$. Moreover, $\blambda_l$ is the $l$th column of $\bLambda$. The optimization problem in (\ref{caponoptim}) has a closed form solution given by \cite{van2002optimum}
\begin{equation}
\bw_l = \frac{\bR_C^{-1}\blambda_l}{\blambda_l^H\bR_C^{-1}\blambda_l}.
\end{equation}
We construct a clutter suppression matrix via concatenating the weighting vectors of all angle cells i.e., $\bW = [\bw_1,...,\bw_L]\in \mathbb{Z}^{M_1\times L}$.
We then apply the clutter suppression matrix to the compressed measurement vector of \eqref{equ:cmvector}, obtaining the clutter and noise mitigated data as
\begin{equation}
\begin{aligned}
\by &= \bW^H \bLambda \bbeta + \bW^H \bxi = \bTheta \bbeta + \bnu,
\end{aligned}
\end{equation}
where $\bTheta$ is the dictionary matrix with $\btheta_l$ corresponding to the $l$th angle cell, and $\bnu$ is the clutter and noise residuals after applying the Capon beamformer with covariance matrix $\bR_T = \bW^H \bR_C \bW$.

\begin{remark}
	{
	When the radar platform is moving, prior observations can lead to models for clutter and consequently $\bR_C$. Our estimate for $\bR_C$ can also be updated in time, based on the received measurements. More precisely, in applications that the clutter-plus-noise covariance matrix is changing smoothly, we can apply a gradually updating technique such as clutter map \cite{4104231} on an ordinary covariance matrix estimation method.	
	Since $\bw_l$ is dependent on the clutter statistics, a change in $\bR_C$ necessitates the re-calculation of $\bw_l$. However, based on the way that $\bR_C$ is changing, it might be possible to update its inverse using the matrix inversion lemma and therefore re-calculate $\bw_l$ with a low computational complexity method. In this paper, the analysis and simulations are provided for a static case (i.e., one snapshot), while a thorough analysis is required to study a dynamic scenario.}
\end{remark}

The spare nature of the Capon beamformer output allows us to achieve further sample reduction by employing another compression matrix $\bPhi_{(2)}\in \mathbb{R}^{M_2\times L}$ as follows \cite{yu2013capon}
\begin{equation}
\begin{aligned}
\bz &= \bPhi_{(2)} \bTheta \bbeta + \bPhi_{(2)}\bnu,
\label{equ:z}
\end{aligned}
\end{equation}
where $\text{CR}_2 = \frac{L}{M_2}$ is the second compression ratio.

%The proposed detection and estimation will operate on $\bz$ and will not attempt to recover $\bbeta$. 
The main problem is now reformulated as determining $\bbeta$ based on the data vector $\bz$ containing the nuisance term $\bPhi_{(2)}\bnu$ (Figure \ref{fig:diag}).
Since we just consider a single range cell, we can restrict our attention to scenarios with a low number of targets. First, we present the single-target scenario, propose detection and angle estimation algorithms, and provide mathematical analytics for the ROC of the proposed detector. Then, we proceed to the more realistic multi-target scenario. We subsequently discuss the inter-relation of single and multi-target scenarios, in order to properly extend the proposed single-target algorithm based on such a relationship.

\subsection{Single-Target Scenario}
In this part, we restrict our attention to the single-target case, i.e., $Q=1$. 
%Therefore, (\ref{equ:ktarget}) can be rewritten as follows:
%
%\begin {equation}
%\bx(n) = \alpha \bb(\theta)\ba^H(\theta)\bs(n) + \bepsilon(n).
%\label{equ:1target}
%\end {equation}
 We can write the hypothesis test based on the data vector in the following form
\begin{equation}
\bz = 
\begin{cases}
\bPhi_{(2)} (\alpha \btheta_t + \bnu), &\mathcal{H}_1: \mbox{if a target exists,}\\
\bPhi_{(2)} {\bnu},  &\mathcal{H}_0: \mbox{otherwise,}
\end{cases}
\end{equation}
where $\alpha$ and $t\in\{1,\ldots,L\}$ are the unknown target's amplitude and index of angle cell, respectively. If a target exists, its angle cell is not known a priori, hence the usual likelihood ratio test (LRT) cannot be computed and used for detection. Instead, we will use the GLRT, in which the LRT is maximized over all grid angles to find the optimum angle cell. The LRT value at the optimal point should then be compared with a proper threshold to test if a target exists or not. 

To determine the threshold, in the following we compute the probability density function (PDF) of the compressed measurement vector for the two hypotheses. Since $\bnu$ has a complex Gaussian distribution with covariance matrix ${{\bR}}_T$, conditioned on $\mathcal{H}_0$, $\bz$ is a vector with distribution $\mathcal{CN}(0,\bPhi_{(2)} {\bR_T} \bPhi_{(2)}^T)$. Thus, for the null hypothesis, we have \cite[p.~258]{papoulis2002probability}

\begin{equation}
\begin{aligned}
&f(\bz|\mathcal{H}_0) = \frac{1}{\pi^{M_2}|\bA|}\exp{\left(-\bz^H\bA^{-1}\bz\right)},
\end{aligned}
\label{equ:H_0}
\end{equation}
where
\begin{equation}
\bA = \bPhi_{(2)} {\bR}_T \bPhi_{(2)}^T.
\end{equation}
For hypothesis $\mathcal{H}_1$, the PDF of $\bz$ conditioned on $\alpha$ and $t$ is given by
\begin{equation}
\begin{aligned}
f(\bz|\mathcal{H}_1,\alpha,t) &= \mathcal{CN}(\alpha \bPhi_{(2)} \btheta_t,\bA)
\\&= \frac{1}{\pi^{M_2}|\bA|}\exp(-(\bz-\alpha \bPhi_{(2)} \btheta_t)^H\\&\bA^{-1}(\bz-\alpha \bPhi_{(2)} \btheta_t)),
\end{aligned}
\label{target1}
\end{equation}
and as mentioned before, the PDF of $\alpha$ is $f(\alpha) = \frac{1}{\pi \sigma_{\alpha}^2} \exp\left(-\frac{|\alpha|^2}{\sigma_{\alpha}^2}\right)$. Consequently, the PDF of $\bz$ under $\mathcal{H}_1$ conditioned on $t$ is derived in the following theorem.

\begin{theorem}
\label{theorem:f(H1)}
The PDF of $\bz$ under $\mathcal{H}_1$ conditioned on $t$ is
\begin{equation}
\begin{aligned}
f(\bz|\mathcal{H}_1,t) &=  \frac{1}{\pi^L|\bA|}\frac{1}{\sigma_{\alpha}^2 d_t+1}\exp{\left(-\bz^H\bA^{-1}\bz\right)} \\&\exp\left(\frac{|e_t|^2\sigma_{\alpha}^2}{\sigma_{\alpha}^2 d_t + 1}\right),
\end{aligned}
\label{equ:H_1}
\end{equation}
where we have defined
\begin{equation}
\begin{aligned}
d_t &= \btheta_t^H \bPhi_{(2)}^T \bA^{-1} \bPhi_{(2)} \btheta_t,\\
e_t &= \btheta^H_t \bPhi_{(2)}^T \bA^{-1} \bz.
\label{equ:definitions}
\end{aligned}
\end{equation}
\end{theorem}

\begin{proof}
The proof is derived in Appendix \ref{proof 1}.
\end{proof}

Knowing both PDFs of $\bz$ under $\mathcal{H}_0$ and $\mathcal{H}_1$ from (\ref{equ:H_0}) and (\ref{equ:H_1}), respectively, 
the LRT can be derived as follows
\begin{equation}
{\rm{L}}(\bz|t) = \frac{f(\bz|\mathcal{H}_1,t)}{f(\bz|\mathcal{H}_0)} = \frac{1}{d_t\sigma_{\alpha}^2+1} \exp\left(\frac{|e_t|^2\sigma_{\alpha}^2}{d_t\sigma_{\alpha}^2+1}\right).
\end{equation}
In order to find the GLRT, ${\rm{L}}(\bz|t)$ should be maximized over $t$,
\begin{equation}
{\rm{GLRT}}(\bz) = \underset{t\in\{1,\ldots,L \}}{\text{max}}~\rm{L}(\bz|t).
\end{equation}
As explained in Appendix \ref{proof 2}, the GLRT can then be obtained as

\begin{equation}
\begin{aligned}
{\rm{GLRT}}(\bz) =  |e_{\hat{t}}| \gtrless \eta,
\end{aligned}
\label{equ:GLRT}
\end{equation}
where $\hat{t} = \underset{t\in\{1,\ldots,L \}}{\arg \max}~{\rm{L}}(\bz|t)$ is the estimation of the index of the target's angle cell and $\eta$ is defined as the detection threshold. Since we employ the Neyman-Pearson detector \cite{poor2013introduction}, $\eta$ is determined based on the desired false alarm probability $P_{fa}$.

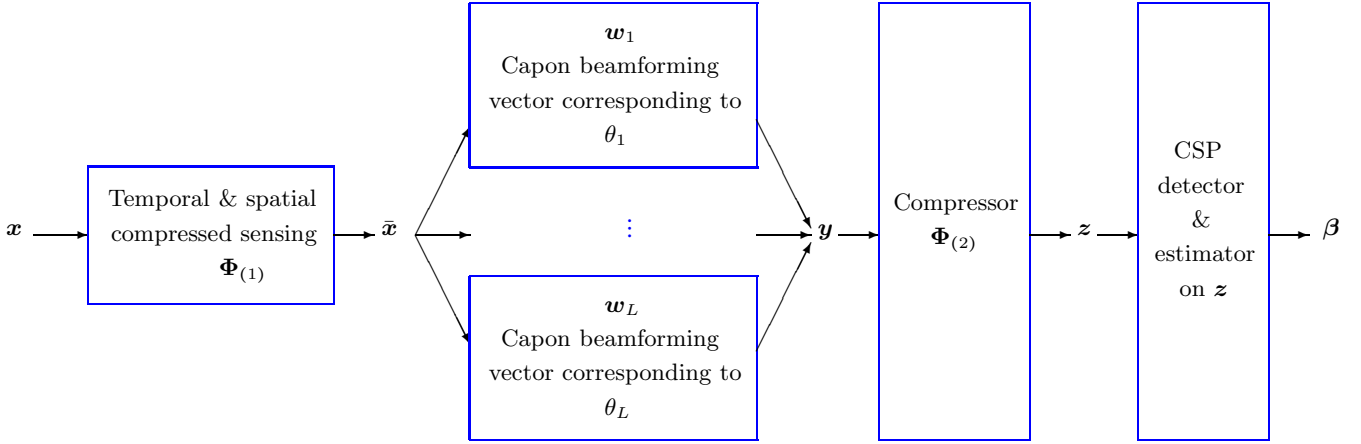
\begin{figure*}
	\setlength{\unitlength}{\textwidth}
	\centering
	\begin{picture}(0.97,0.4)
	\put(0.0,0.2){\color{black}\small $\bx$}	
	\put(0.02,0.2){\vector(1,0){0.04}}
	\put(0.06,0.15){\color{blue}\framebox(0.18,0.1){\color{black}\small$\begin{aligned}{\rm Temporal~} &{\rm \& ~spatial} \\ {\rm compresse} &{\rm d~sensing} \\ &\bPhi_{(1)} \end{aligned}$}}
	%\thicklines
	\thinlines
	\put(0.24,0.2){\vector(1,0){0.03}}
	\put(0.275,0.2){\color{black}\small $\bar{\bx}$}
	\put(0.30,0.2){\vector(1,2){0.04}}
	\put(0.30,0.2){\vector(1,0){0.04}}
	\put(0.30,0.2){\vector(1,-2){0.04}}
	\put(0.34,0.25){\color{blue}\framebox(0.21,0.12){\color{black}\small$\begin{aligned}&\bw_1 \\{\rm Capon~ be} &{\rm amforming} \\ {\rm vector~ corr}&{\rm esponding~ to} \\  &\theta_1  \end{aligned}$}}
	\put(0.455,0.2){\color{blue} $\vdots$}
	\put(0.34,0.05){\color{blue}\framebox(0.21,0.12){\color{black}\small$\begin{aligned}&\bw_L \\{\rm Capon~ be} &{\rm amforming} \\ {\rm vector~ corr}&{\rm esponding~ to} \\  &\theta_L  \end{aligned}$}}
	\put(0.55,0.285){\vector(1,-2){0.04}}
	\put(0.55,0.2){\vector(1,0){0.04}}
	\put(0.55,0.115){\vector(1,2){0.04}}
	\put(0.595,0.2){\color{black}\small $\by$}
	\put(0.61,0.2){\vector(1,0){0.03}}
	\put(0.64,0.05){\color{blue}\framebox(0.11,0.32){\color{black}\small$\begin{aligned}{\rm Compre}&{\rm ssor} \\ \bPhi_{(2)} \end{aligned}$}}
	\put(0.75,0.2){\vector(1,0){0.03}}
	\put(0.785,0.2){\color{black}\small $\bz$}
	\put(0.8,0.2){\vector(1,0){0.03}}
	\put(0.83,0.05){\color{blue}\framebox(0.095,0.32){\color{black}\small$\begin{aligned}{\rm C}&{\rm SP} \\ {\rm det}&{\rm ector} \\ &\& \\ {\rm esti}&{\rm mator}\\ {\rm o} & {\rm n}~\bz \end{aligned}$}}
	\put(0.925,0.2){\vector(1,0){0.03}}
	\put(0.965,0.2){\color{black}\small $\bbeta$}
	\end{picture}
	\caption{A high level overview of the proposed method}
	\label{fig:diag}
\end{figure*}

\subsubsection{Detector Performance}
Here, we analyze the ROC of the proposed detector. It holds that
\begin{equation}
%\begin{aligned}
e_{\hat{t}}|\mathcal{H}_0 \sim \mathcal{CN}(0,d_{\hat{t}}), \\
\end{equation}
and the false alarm probability equals
\begin{equation}
\begin{aligned}
P_{fa} &= p(|e_{\hat{t}}|>\eta|\mathcal{H}_0) = \int_{\eta}^{\infty} f(x|\mathcal{H}_0)dx \\&= 1-\int_{-\infty}^{\eta} f(x|\mathcal{H}_0)dx = \exp(-\frac{\eta^2}{d_{\hat{t}}}),
\end{aligned}
\label{equ:pfa}
\end{equation}
where we introduce $x = |e_{\hat{t}}|$ and hence $x|\mathcal{H}_0 \sim  {\rm Rayleigh}(\gamma_0)$ with $\gamma_0=(\frac{d_{\hat{t}}}{2})^{0.5}$.

Also, it holds that
\begin{equation}
\begin{aligned}
e_{\hat{t}}|\mathcal{H}_1,\alpha &\sim \mathcal{CN}(\alpha d_{\hat{t}},d_{\hat{t}}), \\
f(e_{\hat{t}}|\mathcal{H}_1) &= \int f(e_{\hat{t}}|\mathcal{H}_1,\alpha)f(\alpha)d\alpha.
\end{aligned}
\label{equ:calfH1}
\end{equation}
We work out (\ref{equ:calfH1}) in Appendix \ref{proof 3} and show that the PDF of $e_{\hat{t}}$ conditioned on $\mathcal{H}_1$ is given by

\begin{equation}
\begin{aligned}
e_{\hat{t}}|\mathcal{H}_1 \sim \mathcal{CN}(0,d_{\hat{t}} + d_{\hat{t}}^2 \sigma_{\alpha}^2).
\end{aligned}
\label{equ:hH_1}
\end{equation}
The detection probability thus equals
\begin{equation}
\begin{aligned}
P_d &= p(|e_{\hat{t}}|>\eta|\mathcal{H}_1) = \int_{\eta}^{\infty} f(x|\mathcal{H}_1)dx \\&= 1-\int^{\eta}_{-\infty} f(x|\mathcal{H}_1)dx = \exp(-\frac{\eta^2}{d_{\hat{t}}+d_{\hat{t}}^2 \sigma_{\alpha}^2}),
\end{aligned}
\end{equation}
where $x|\mathcal{H}_1 \sim {\rm Rayleigh}(\gamma_1)$ with $\gamma_1=(\frac{d_{\hat{t}}+d_{\hat{t}}^2\sigma_{\alpha}^2}{2})^{0.5}$.

Rearranging (\ref{equ:pfa}), it is possible to obtain $\eta$ from $P_{fa}$ as
\begin{equation}
\eta^2 = -d_{\hat{t}}\ln P_{fa}.
\end{equation}
Thus, the ROC equation is obtained as
\begin{equation}
P_d = \exp(\frac{d_{\hat{t}}\ln P_{fa}}{d_{\hat{t}}+d_{\hat{t}}^2\sigma_{\alpha}^2}) = P_{fa}^{(1+d_{\hat{t}}\sigma_{\alpha}^2)^{-1}},
\label{equ:ROC}
\end{equation}
where $P_d$ is derived as a function of $P_{fa}$.

\subsubsection{Measurements SNCR at Input and Output}

In this part, we analyze the signal to noise and clutter ratio (SNCR) both at input and output of the detector.

For input, we need to calculate the SNCR for equation \eqref{regularmeasures}. The signal power is $\mathbb{E}\{||\alpha\bpsi_t||_2^2\}=\sigma_{\alpha}^2RNP$, where $P$ is the transmit power. Also, the noise plus clutter power is $\mathbb{E}\{||\bepsilon||_2^2\}=\rm{Tr}(\tilde{\bR}_N)$. Therefore, the input SNCR is
\begin{equation}
\rm{SNCR}_{\rm{in}}=\frac{\sigma_{\alpha}^2RNP}{\rm{Tr}(\tilde{\bR}_N)}.
\end{equation}
Based on \eqref{equ:GLRT}, the statistic for GLRT is
\begin{equation}
\begin{aligned}
x = |e_{\hat{t}}| = |\btheta^H_{\hat{t}} \bPhi_{(2)}^T \bA^{-1} \bz|.
\end{aligned}
\end{equation}
In Appendix \ref{proof 4}, we derive an approximation of $x$ in the following form
\begin{equation}
\begin{aligned}
x = |\alpha| d_t+\Re\bigg\{\frac{|\alpha|}{\alpha}\btheta^H_t \bPhi_{(2)}^T \bA^{-1}\bPhi_{(2)}\bnu\bigg\}.
\end{aligned}
\label{derxmidmain}
\end{equation}
Denoting $\alpha$ in the polar form as $\alpha=|\alpha|\rm{e}^{\rm{j}\omega}$, \eqref{derxmidmain} can be simplified as
\begin{equation}
x = |\alpha| d_t+\Re\bigg\{\rm{e}^{-\rm{j}\omega} \btheta^H_t \bPhi_{(2)}^T \bA^{-1}\bPhi_{(2)}\bnu\bigg\}
\label{xformfinal}
\end{equation}
In \eqref{xformfinal}, the signal term is $|\alpha| d_t$, while the noise plus clutter term is $\Re\bigg\{\rm{e}^{-\rm{j}\omega} \btheta^H_t \bPhi_{(2)}^T \bA^{-1}\bPhi_{(2)}\bnu\bigg\}$. Calculation of signal power is straightforward and is equal to $\sigma_{\alpha}^2d_t^2$. For the noise and clutter term, defining $g\triangleq\btheta^H_t \bPhi_{(2)}^T \bA^{-1}\bPhi_{(2)}\bnu$, the noise plus clutter term can be expressed as
\begin{equation}
\begin{aligned}
\Re\bigg\{\rm{e}^{-\rm{j}\omega} \btheta^H_t \bPhi_{(2)}^T \bA^{-1}\bPhi_{(2)}\bnu\bigg\} &= \Re\{g\rm{e}^{-\rm{j}\omega}\} \\&= g_r\cos\omega + g_i\sin\omega,
\end{aligned}
\end{equation}
where $g_r$ and $g_i$ are the real and imaginary parts of $g$, respectively. Subsequently, the noise plus clutter power is given by
\begin{equation}
\begin{aligned}
\mathbb{E}\{(g_r&\cos\omega + g_i\sin\omega)^2\} \\&= \mathbb{E}\{g_r^2\cos^2\omega+g_i^2\sin^2\omega+2g_rg_i\cos\omega\sin\omega\}\\&=\mathbb{E}\{g_r^2\cos^2\omega\},
\end{aligned}
\end{equation}
since $\bnu$ is circular normal, its real and imaginary parts are statistically independent and both are zero-mean. In addition, $g$ and $\omega$ are independent and $\omega$ has a uniform distribution in the interval $[0,2\pi]$. Thus, noise plus clutter power $P_c$ is calculated as follows
\begin{equation}
P_c = \frac{1}{2}\mathbb{E}\{g_r^2+g_i^2\}=\frac{1}{2}\mathbb{E}\{g^2\} = \frac{1}{2}d_t.
\end{equation}
As a consequence, we have
\begin{equation}
\rm{SNCR}_{\rm{out}}=2\sigma_{\alpha}^2d_t,
\end{equation}
and using \eqref{equ:ROC}, the SNCR is obtained as
\begin{equation}
\rm{SNCR}_{\rm{out}}=2\left(\frac{\ln P_{fa}}{\ln P_d}-1\right).
\end{equation}

\subsection{Multi-Target Scenario}
To address the multi-target scenario, we can use a deflation type approach for detecting one target at a time, along the lines of \cite{4385788}. In each iteration, using the single-target algorithm, the strongest target is extracted. If the target is greater than the threshold $\eta$, the target is detected, and its contribution is eliminated from the compressed measurement vector $\bz$. Then, iterations continue with the residual of the measurement vector, until no target is detected. The idea of residual updating is similar to the procedure done in orthogonal matching pursuit (OMP) \cite{4385788}.

The multi-target detection algorithm proceeds as follows.

\begin{enumerate}
%	\enumerate
	\item Initialization:
	\begin{equation}
	\begin{aligned}	
	\bz_{(r)} &= \bz,\\
	P &= 0,\\
	\mathcal{A} &= \emptyset,
	\end{aligned}
	\end{equation}
	\item Estimation:
	\begin{equation}
	\begin{aligned}
	\widehat{t}_{P+1} &= \underset{t\in\{1,\ldots,L \}}{\arg \max}~{\rm{L}}(\bz_{(r)}|t),\\	
	\end{aligned}
	\end{equation}
	\item Detection:
	\begin{equation}
	{\rm{GLRT}}(\bz_{(r)}) =  |e_{\widehat{t}_{P+1}}| \gtrless \eta,
	\end{equation}
	If no target is detected, terminate the algorithm.
	\item Residual updating:
	\begin{equation}
	\begin{aligned}
	P &= P+1,\\
	\mathcal{A} &= \mathcal{A} \cup \{\widehat{t}_{P+1}\},\\
%	\mathcal{A} &= \mathcal{A} \cup \{\widehat{t}_q \},\\
	\widehat{\balpha} &= (\bTheta[\mathcal{A}]^H\bPhi_{(2)}^T\bPhi_{(2)}\bTheta[\mathcal{A}])^\dagger \bTheta[\mathcal{A}]^H\bPhi_{(2)}^T\bz,\\
	\bz_{(r)} &= \bz - \bPhi_{(2)}\bTheta[\mathcal{A}]\widehat{\balpha},
	\end{aligned}
	\end{equation}
	\item Go to step 2.
\end{enumerate}
Executing this multi-target detector and estimator algorithm, $P$ is the estimated number of targets, ${\mathcal A}$ is the set of detected angle cell indices, $\widehat{\alpha}_p$ and $\widehat{t}_p$, $p=1,...,P$, are the estimated targets' amplitude and angle cell index, respectively. Subsequently, $\bbeta$ is a zero vector except for the entries $\widehat{t}_p$, $p=1,...,P$ which are equal to $\widehat{\alpha}_p$, $p=1,...,P$.

% ********************* simulation results
\section{Simulation Results}
In this section, we present three sets of simulations to evaluate the proposed algorithm from different perspectives. First, the single-target case is considered and four scenarios are simulated. The first scenario studies the effects of the number of antennas and SNR on the performance of the proposed algorithm, while the second scenario is dedicated to the problem of grid mismatch. In the third and fourth scenarios, the proposed algorithm is compared to a state of the art algorithm investigating the effect of the compression ratio and grid mismatch, respectively. The multi-target case is also considered, and the performance of the proposed approach is evaluated for different numbers of targets. Finally, a comparison between CSP-MIMO radar and conventional CS-MIMO radar in terms of saving in sample complexity is presented. {Unless specifically mentioned, targets are assumed to fall on the grid.}

We evaluate the performance of the proposed detection approach through an ROC analysis. We also provide the bias and standard deviation (std) of the proposed angle estimator and investigate the impact of parameters such as the number of antennas, $\rm{SNR}$, and compression ratio, on the performance. Further, we compare the proposed algorithm with one of the state-of-the-art compressed sensing recovery based algorithms, NESTA \cite{becker2011nesta} in terms of ROC, estimation accuracy, and execution time in various scenarios. NESTA is chosen, because it is a fast and accurate sparse recovery algorithm and is shown to perform well on the problem of signal reconstruction in MIMO radar \cite{tohidi2017compressive}. For the multi-target case, the OMP algorithm is also compared with the proposed algorithm. We considered the OMP algorithm for the multi-target scenario because it has a similar residual update as our proposed algorithm. It should be mentioned that the procedure for all the algorithms, i.e., CSP, NESTA, and OMP, is the same and all are performed on $\bz$ (except the last simulation), i.e., \eqref{equ:z}. A Monte Carlo simulation with 10000 runs is employed, where unless mentioned specifically, parameters are selected based on Table \ref{tab:table1}. In the following simulations, without loss of generality, the compression matrices are chosen to be Gaussian with independent identically distributed (i.i.d.) entries having zero mean and unit variance.
Simulations are performed in a MATLAB R2017b environment, using an Intel Core (TM) i7-4790K, 4 GHz processor with 64 GB of memory, and under
a 64 bit Microsoft Windows 10 operating system. We first present simulations for single-target scenarios and then proceed to multi-target scenarios.

\begin{table}[t!]
  \centering
  \caption{Simulation parameters}
  \label{tab:table1}
  \begin{tabular}{|c|c|c|}
  \hline
    Description & Parameter & Value\\
    \hline
    Number of receive antennas & $R$ & 8 \\    
    \hline
    Number of transmit antennas & $I$ & 10 \\    
    \hline
    Number of samples & $N$ & 20 \\
    \hline
    Signal to noise ratio & SNR & 0 dB \\
    \hline
    First compression ratio & $\rm{CR}_1$ & 4 \\
    \hline
    Second compression ratio & $\rm{CR}_2$ & 2 \\
    \hline
    Clutter to noise ratio & CNR & 30dB \\
    \hline
    $\theta$ span & & -50 to +50 \\
    &&degree \\
    \hline
    $\theta$ resolution & & 2 degree \\
    \hline
  \end{tabular}
\end{table}

\subsection{Single-Target Scenario}

Figure \ref{fig:SNR_Mr} presents the ROC and estimation accuracy for different values of SNR and number of receive antennas, $R$. In Figure \ref{fig:SNR_Mra}, the ROC related to the Monte Carlo simulations is plotted. The theoretical calculation in (\ref{equ:ROC}) is also shown where we use the true target location $t$ instead of $\hat{t}$. Therefore, the theoretical ROC should be an upper bound for the simulated ROC, however, it appears to match well the simulations results. In addition, increasing the SNR and $R$ leads to a higher probability of detection for the same false alarm probability. Figure \ref{fig:SNR_Mrb} demonstrates the bias and std of the proposed angle estimation algorithm. Again, a reduction in estimation bias (i.e., approaching the real value) and std when increasing SNR and $R$ is observed.
\begin{figure}
\centering
\psfrag{Pd}{\scriptsize{$P_d$}}
\psfrag{Pfa}{\scriptsize{$P_{fa}$}}
\psfrag{R - SNR (dB)}{\scriptsize{R - SNR (dB)}}
\psfrag{Estimation bias and std}{\scriptsize{Estimation bias and std}}
\psfrag{CSP-R=8-SNR=0dB}{\scriptsize {CSP ~~ - $R$ = 8~ - SNR = -10dB}}
\psfrag{CSP-R=16-SNR=0dB}{\scriptsize {CSP ~~ - $R$ = 16 - SNR = -10dB}}
\psfrag{CSP-R=8-SNR=20dB}{\scriptsize {CSP ~~ - $R$ = 8~ - SNR = 10dB}}
\psfrag{CSP-R=16-SNR=20dB}{\scriptsize {CSP ~~ - $R$ = 16 - SNR = 10dB}}
\psfrag{theory-R=8-SNR=0dB}{\scriptsize {Theory - $R$ = 8~ - SNR = -10dB}}
\psfrag{theory-R=16-SNR=0dB}{\scriptsize {Theory - $R$ = 16 - SNR = -10dB}}
\psfrag{theory-R=8-SNR=20dB}{\scriptsize {Theory - $R$ = 8~ - SNR = 10dB}}
\psfrag{theory-R=16-SNR=20dB}{\scriptsize {Theory - $R$ = 16 - SNR = 10dB}}
\subfigure[] {\includegraphics[width=.44\textwidth]{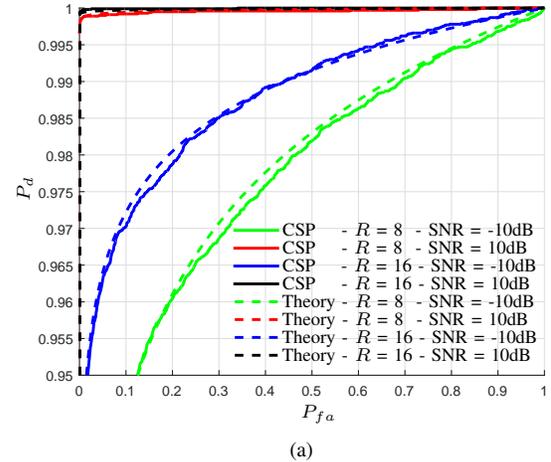}
\label{fig:SNR_Mra}} \quad
\subfigure[] {\includegraphics[width=.44\textwidth]{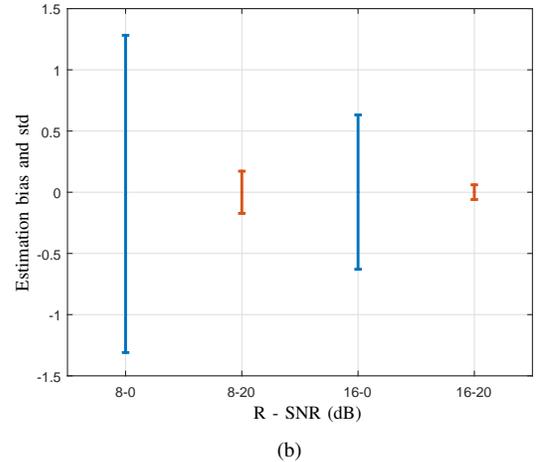}
\label{fig:SNR_Mrb}} \\
\caption{Performance of the proposed CSP approach for different SNR and $R$ values (a) ROC , (b) angle estimation bias and std.}
\label{fig:SNR_Mr}
\end{figure}

For the next scenario, we examine the performance of the proposed approach when the targets do not fall on the grid. The sensitivity of CS-based target estimation methods to grid mismatch has been extensively discussed in the literature \cite{5710590,6576276,6320676,5706373}. Figure \ref{fig:offgrid_Mr} depicts the effect of both $R$ and mismatch values on the ROC and estimation accuracy. %As mentioned in Table \ref{tab:table1}, the angle grid points are $2$ degrees apart and $10$ degrees is one of them. 
We consider a grid mismatch equal to $0.1$ and $1$ degrees. %(i.e., a target at angles $9.9$ and $9$ degrees, respectively).
 It is observed from Figure \ref{fig:offgrid_Mra} that increasing $R$, improves the ROC. Although the ROC is not very sensitive to the mismatch values, an increasing mismatch, results in a slightly lower probability of detection for the same false alarm probability. As demonstrated in Figure \ref{fig:offgrid_Mrb}, the proposed algorithm achieves a higher angle estimation accuracy as the number of antennas increases. In addition, it is depicted that increasing the mismatch value, increases the estimation std.

\begin{figure}
\centering
\psfrag{Pd}{\scriptsize{$P_d$}}
\psfrag{Pfa}{\scriptsize{$P_{fa}$}}
\psfrag{R - mismatch}{\scriptsize{R - Mismatch}}
\psfrag{Estimation bias and std}{\scriptsize{Estimation bias and std}}
\psfrag{R=4-mismatch=0.1}{\scriptsize {$R$ =~ 4 - Mismatch = 0.1}}
\psfrag{R=4-mismatch=1}{\scriptsize {$R$ =~ 4 - Mismatch = 1}}
\psfrag{R=8-mismatch=0.1}{\scriptsize {$R$ =~ 8 - Mismatch = 0.1}}
\psfrag{R=8-mismatch=1}{\scriptsize {$R$ =~ 8 - Mismatch = 1}}
\psfrag{R=16-mismatch=0.1}{\scriptsize {$R$ = 16 - Mismatch = 0.1}}
\psfrag{R=16-mismatch=1}{\scriptsize {$R$ = 16 - Mismatch = 1}}
\subfigure[] {\includegraphics[width=.44\textwidth]{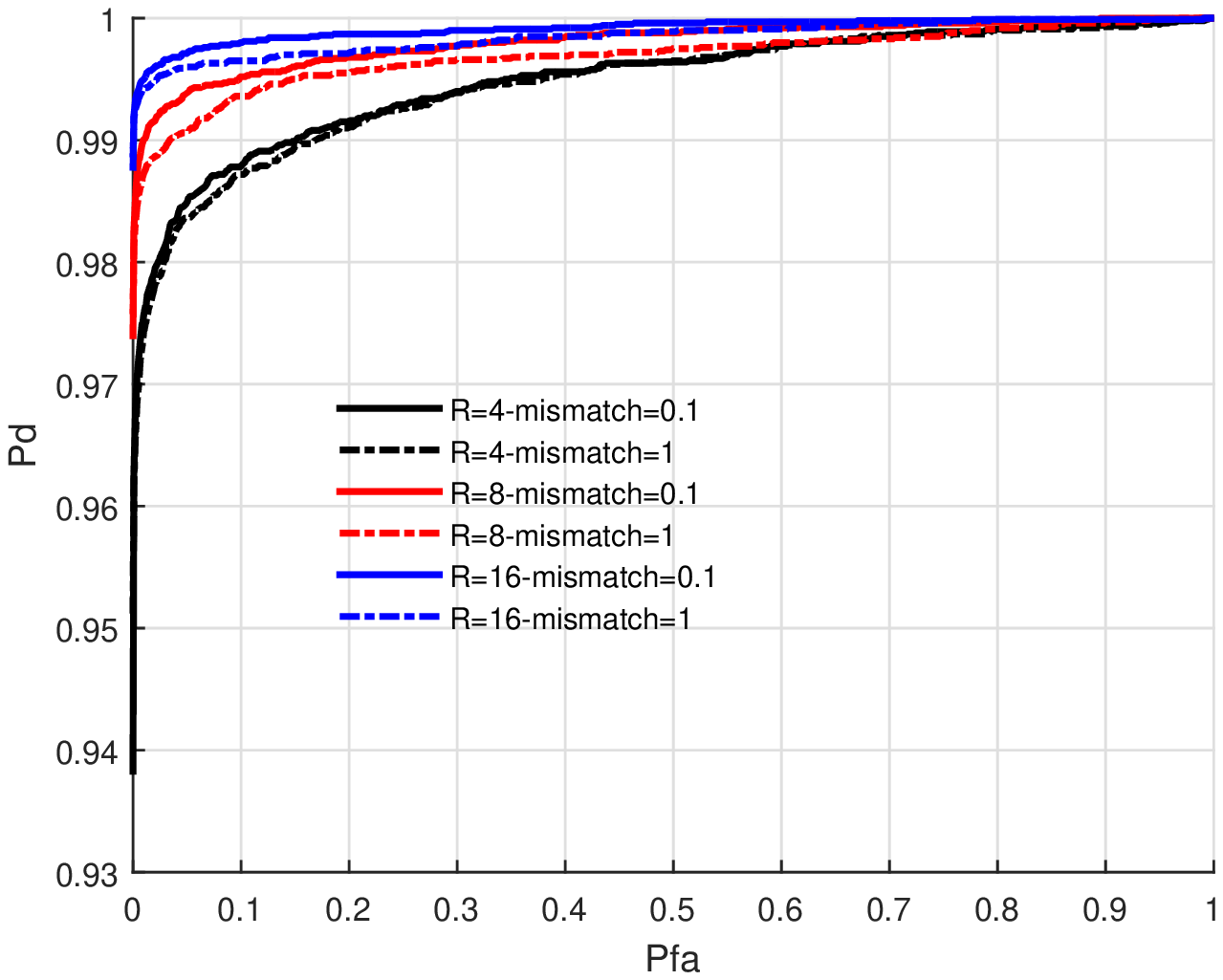}
\label{fig:offgrid_Mra}} \quad
\subfigure[] {\includegraphics[width=.44\textwidth]{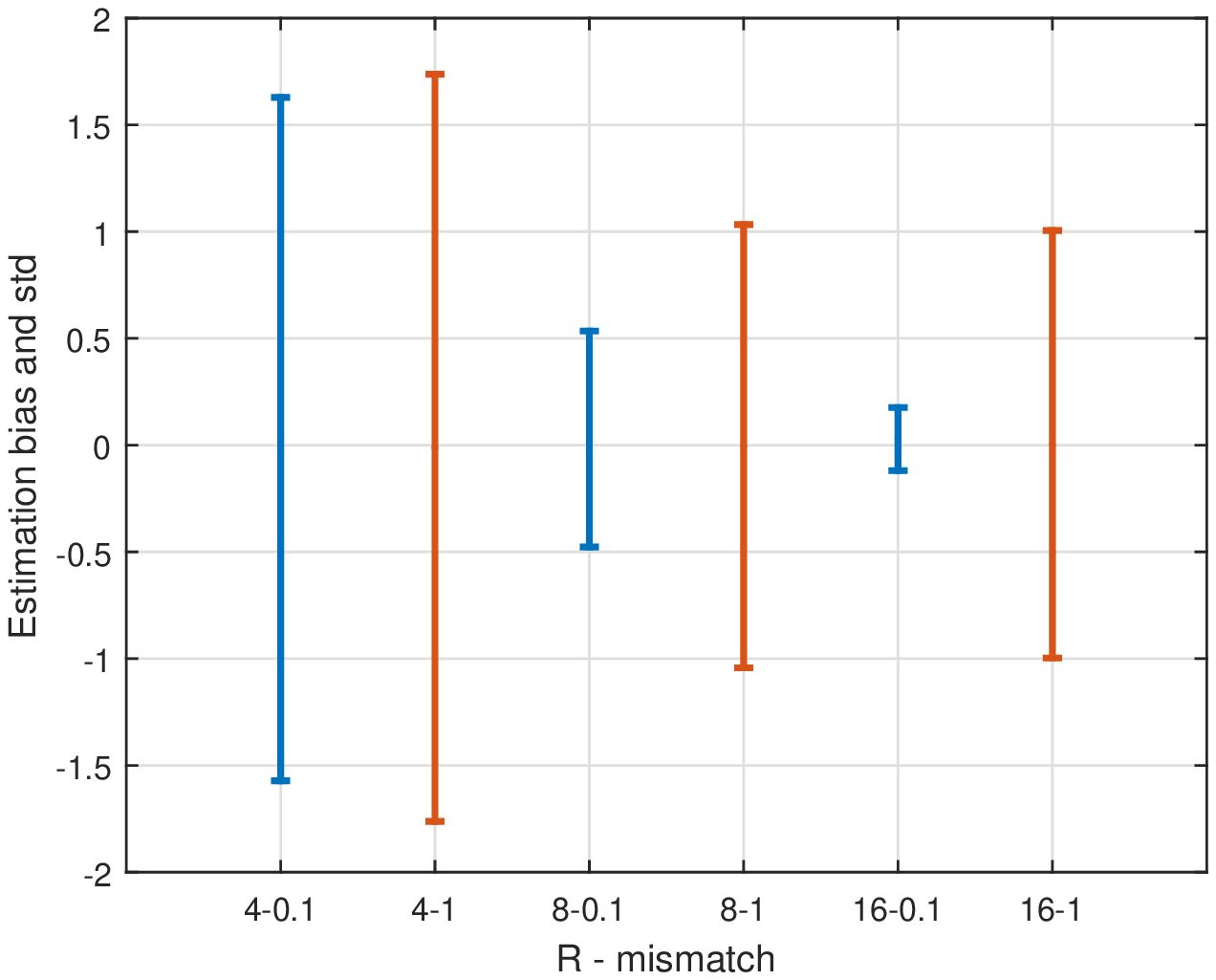}
\label{fig:offgrid_Mrb}} \\
\caption{Performance of the proposed CSP approach for different values of grid mismatch and $R$ (a) ROC , (b) angle estimation bias and std.}
\label{fig:offgrid_Mr}
\end{figure}

As mentioned earlier, some of the important advantages of processing in the compressed domain are the reduced amounts of data, computational complexity, and memory usage in the system. Here, we compare the proposed CSP approach with one of the more recent sparse recovery methods, NESTA, from various aspects. Figure \ref{fig:CR_Nesta} presents the comparison for different values of compression ratio $\text{CR}_{1}$. The superiority of the proposed algorithm in terms of ROC is plotted in Figure \ref{fig:CR_Nestaa}. The theoretical calculations are also shown to be well matched with the simulation results. Moreover, as shown in Figure \ref{fig:CR_Nestab}, the proposed method slightly outperforms NESTA based on angle estimation accuracy. Furthermore, Figure \ref{fig:CR_Nestac} compares the algorithms from the execution time point of view. A significant reduction in the execution time is observed in this figure.

\begin{figure*}
\centering
\psfrag{Pd}{\scriptsize{$P_d$}}
\psfrag{Pfa}{\scriptsize{$P_{fa}$}}
\psfrag{CR}{\scriptsize{$\rm{CR}_1$}}
\psfrag{Time (s)}{\scriptsize{Time (s)}}
\psfrag{CSP}{\tiny{CSP}}
\psfrag{NESTA}{\tiny{NESTA}}
\psfrag{Estimation bias and std}{\scriptsize{Estimation bias and std}}
\psfrag{CSP-CR=1}{\tiny {CSP~~~~ - $\rm{CR}_1$ = 1}}
\psfrag{CSP-CR=2}{\tiny {CSP~~~~ - $\rm{CR}_1$ = 2}}
\psfrag{CSP-CR=4}{\tiny {CSP~~~~ - $\rm{CR}_1$ = 4}}
\psfrag{CSP-CR=8}{\tiny {CSP~~~~ - $\rm{CR}_1$ = 8}}
\psfrag{CSP-CR=16}{\tiny {CSP~~~~ - $\rm{CR}_1$ = 16}}
\psfrag{NESTA-CR=1}{\tiny {NESTA - $\rm{CR}_1$ = 1}}
\psfrag{NESTA-CR=4}{\tiny {NESTA - $\rm{CR}_1$ = 4}}
\psfrag{NESTA-CR=16}{\tiny {NESTA - $\rm{CR}_1$ = 16}}
\psfrag{theory-CR=1}{\tiny {Theory - $\rm{CR}_1$ = 1}}
\psfrag{theory-CR=4}{\tiny {Theory - $\rm{CR}_1$ = 4}}
\psfrag{theory-CR=16}{\tiny {Theory - $\rm{CR}_1$ = 16}}
\subfigure[] {\includegraphics[width=.3\textwidth]{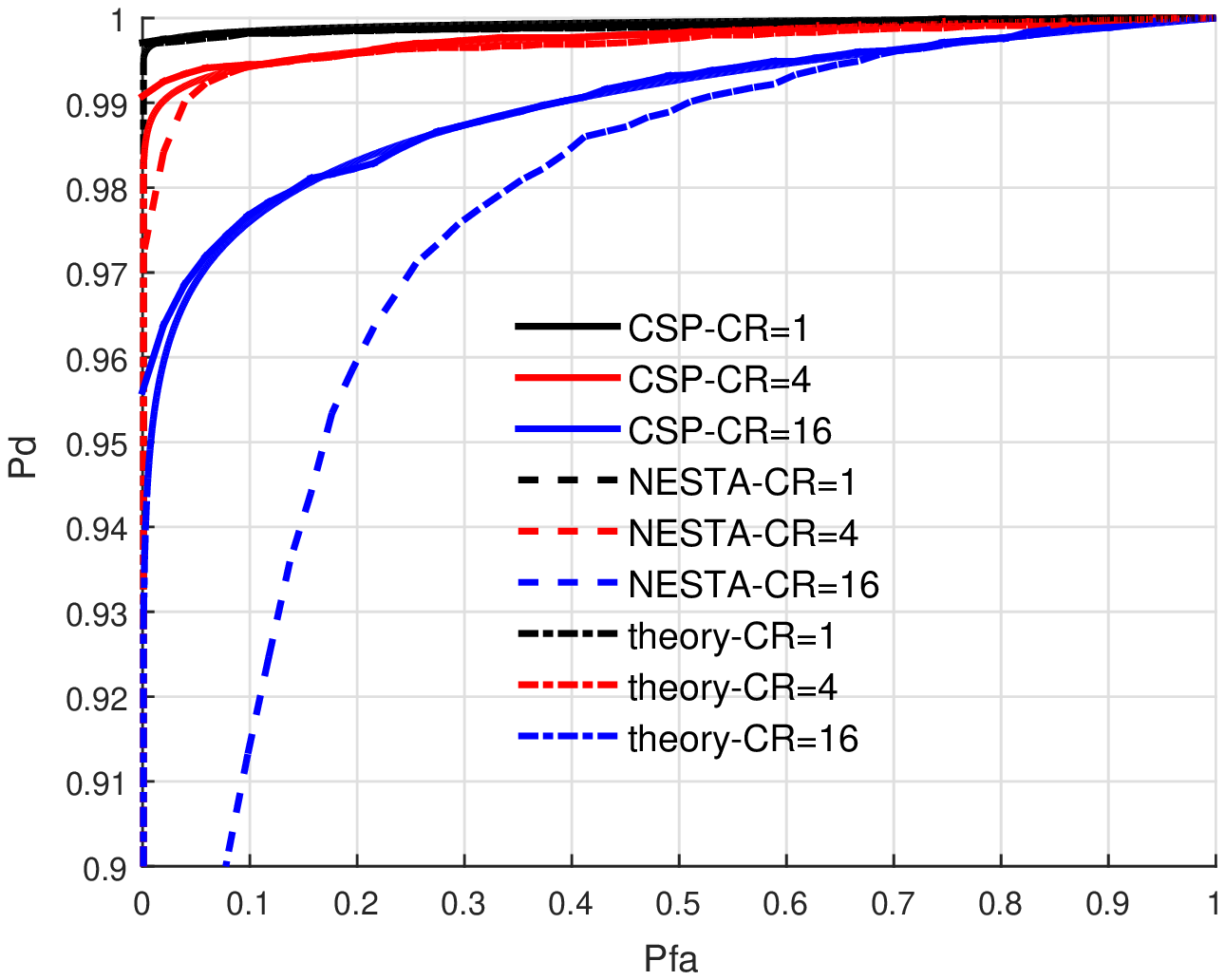}
\label{fig:CR_Nestaa}} \quad
\subfigure[] {\includegraphics[width=.3\textwidth]{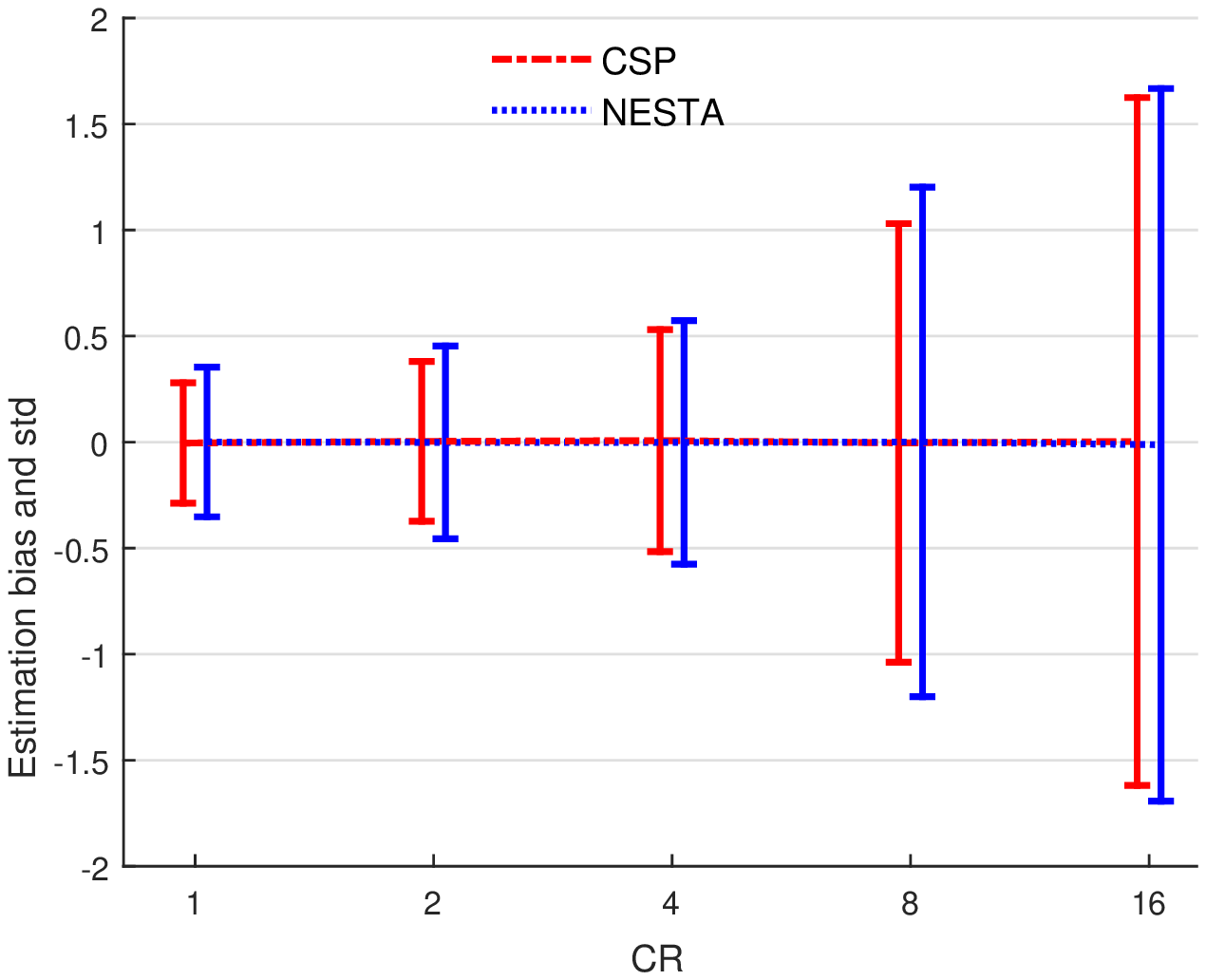}
\label{fig:CR_Nestab}} \quad
\subfigure[] {\includegraphics[width=.3\textwidth]{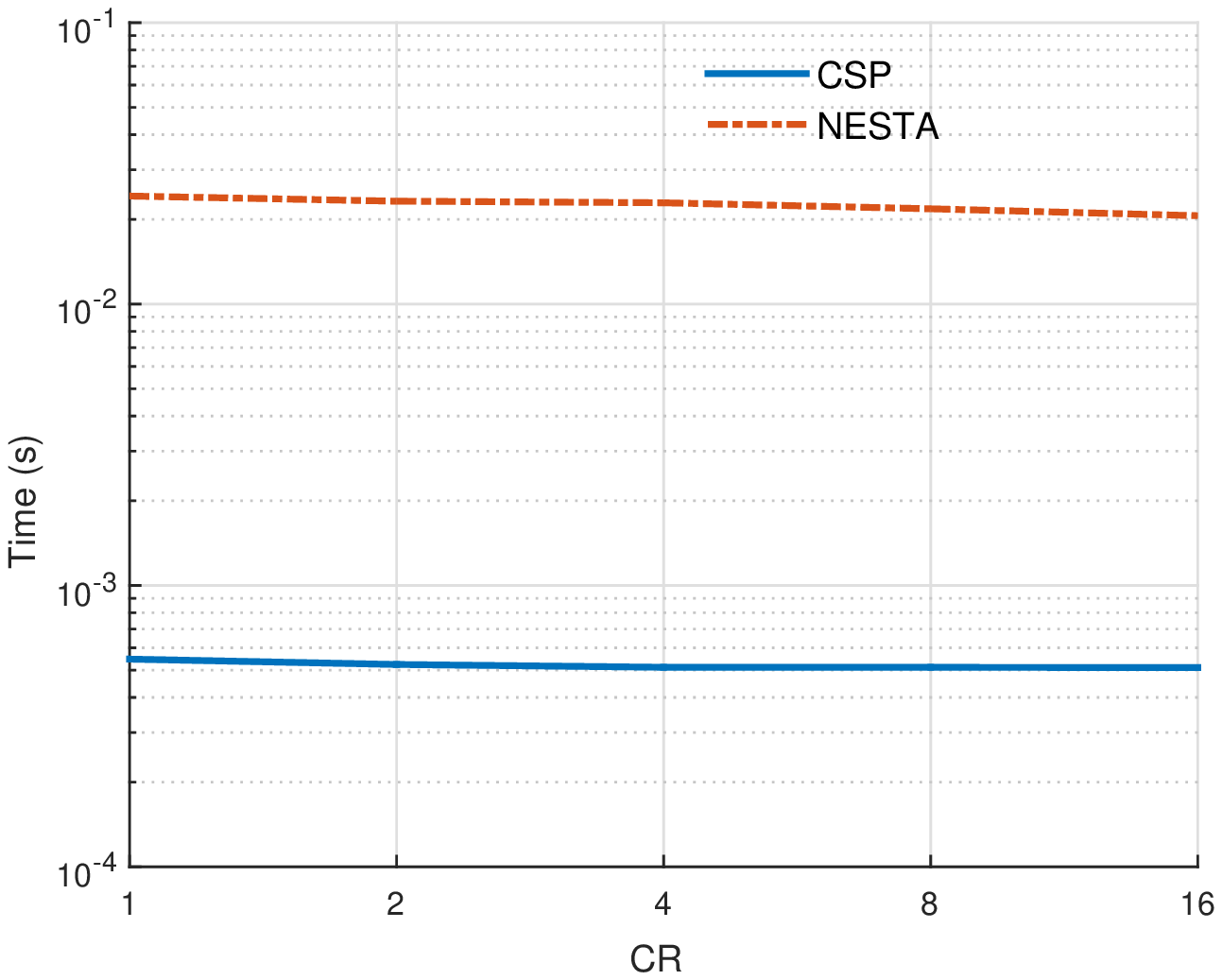}
\label{fig:CR_Nestac}} \\
\caption{Comparison of the proposed CSP approach with NESTA for different values of $\rm{CR}_1$ (a) ROC , (b) angle estimation bias and std, (c) execution time.}
\label{fig:CR_Nesta}
\end{figure*}

Figure \ref{figure:offgrid_Nesta} presents a comparison with NESTA for scenarios with grid mismatch. In this case, the target grid mismatch equals $0.1$, $0.3$, $0.5$, and $1$ degrees. %(i.e., target at angles $9.9$, $9.7$, $9.5$, and $9$ degrees, respectively). 
Although the ROC curves of both methods are close, the proposed algorithm achieves better ROC curves, which are plotted in Figure \ref{figure:offgrid_Nestaa}. Also, an improved estimation accuracy of the proposed method for different values of mismatch is depicted in Figure \ref{figure:offgrid_Nestab}.
%Besides the better ROC curves of the proposed algorithm, which are plotted in Figure \ref{figure:offgrid_Nestaa}, an improved estimation accuracy for different values of mismatch is achieved. %Further, the proposed algorithm results in unbiased angle estimation while the NESTA estimates exhibit a bias. Also, the error std of the proposed estimation method is lower than NESTA's.

\begin{figure}
\centering
\psfrag{Pd}{\scriptsize{$P_d$}}
\psfrag{Pfa}{\scriptsize{$P_{fa}$}}
\psfrag{mismatch}{\scriptsize{Mismatch}}
\psfrag{CSP}{\scriptsize{CSP}}
\psfrag{NESTA}{\scriptsize{NESTA}}
\psfrag{Estimation bias and std}{\scriptsize{Estimation bias and std}}
\psfrag{CSP-mismatch=0.1}{\scriptsize {CSP~~~~ - Mismatch = 0.1}}
\psfrag{CSP-mismatch=0.3}{\scriptsize {CSP~~~~ - Mismatch = 0.3}}
\psfrag{CSP-mismatch=0.5}{\scriptsize {CSP~~~~ - Mismatch = 0.5}}
\psfrag{CSP-mismatch=1}{\scriptsize {CSP~~~~ - Mismatch = 1}}
\psfrag{NESTA-mismatch=0.1}{\scriptsize {NESTA - Mismatch = 0.1}}
\psfrag{NESTA-mismatch=0.3}{\scriptsize {NESTA - Mismatch = 0.3}}
\psfrag{NESTA-mismatch=0.5}{\scriptsize {NESTA - Mismatch = 0.5}}
\psfrag{NESTA-mismatch=1}{\scriptsize {NESTA - Mismatch = 1}}
\subfigure[] {\includegraphics[width=.44\textwidth]{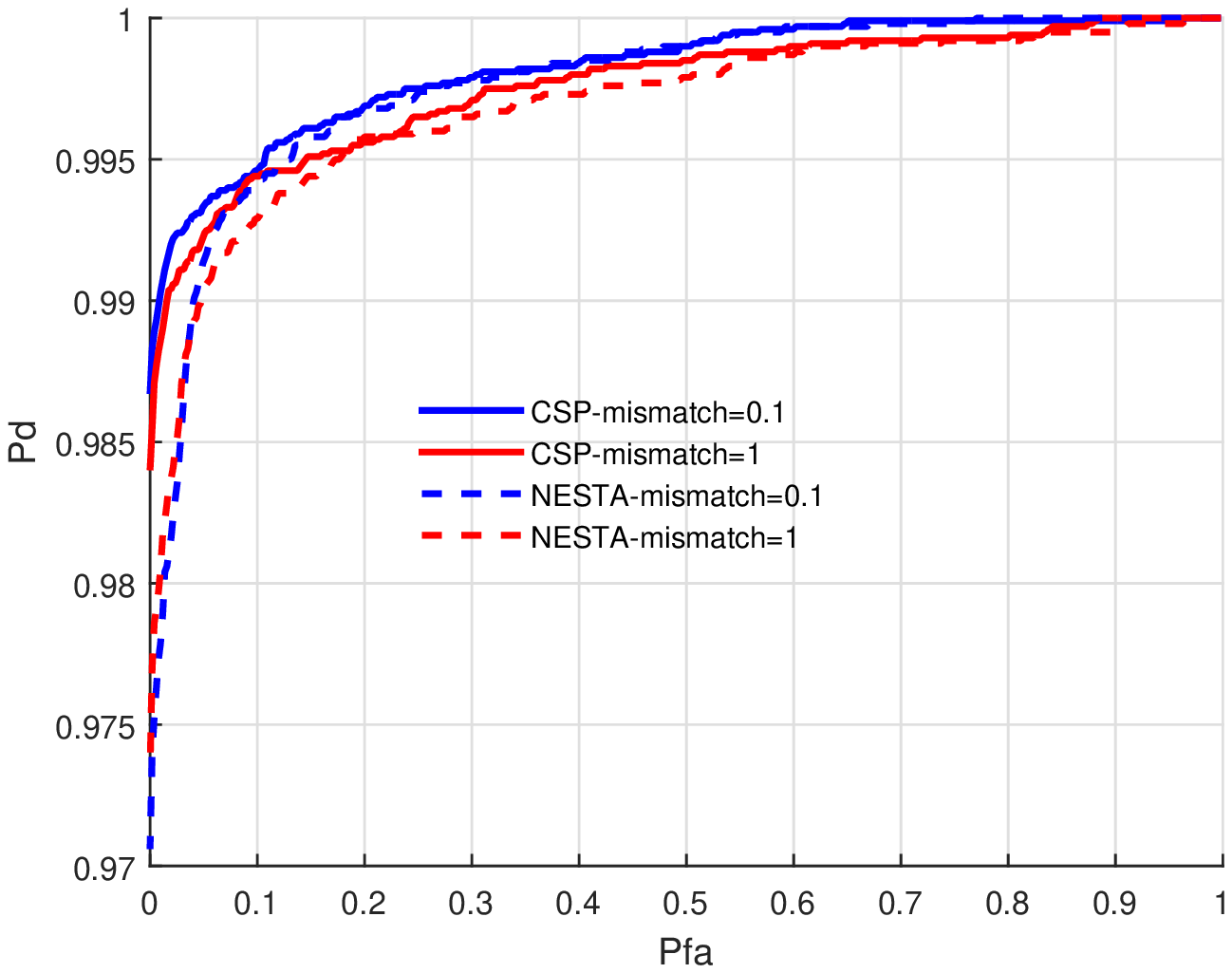}
\label{figure:offgrid_Nestaa}} \quad
\subfigure[] {\includegraphics[width=.44\textwidth]{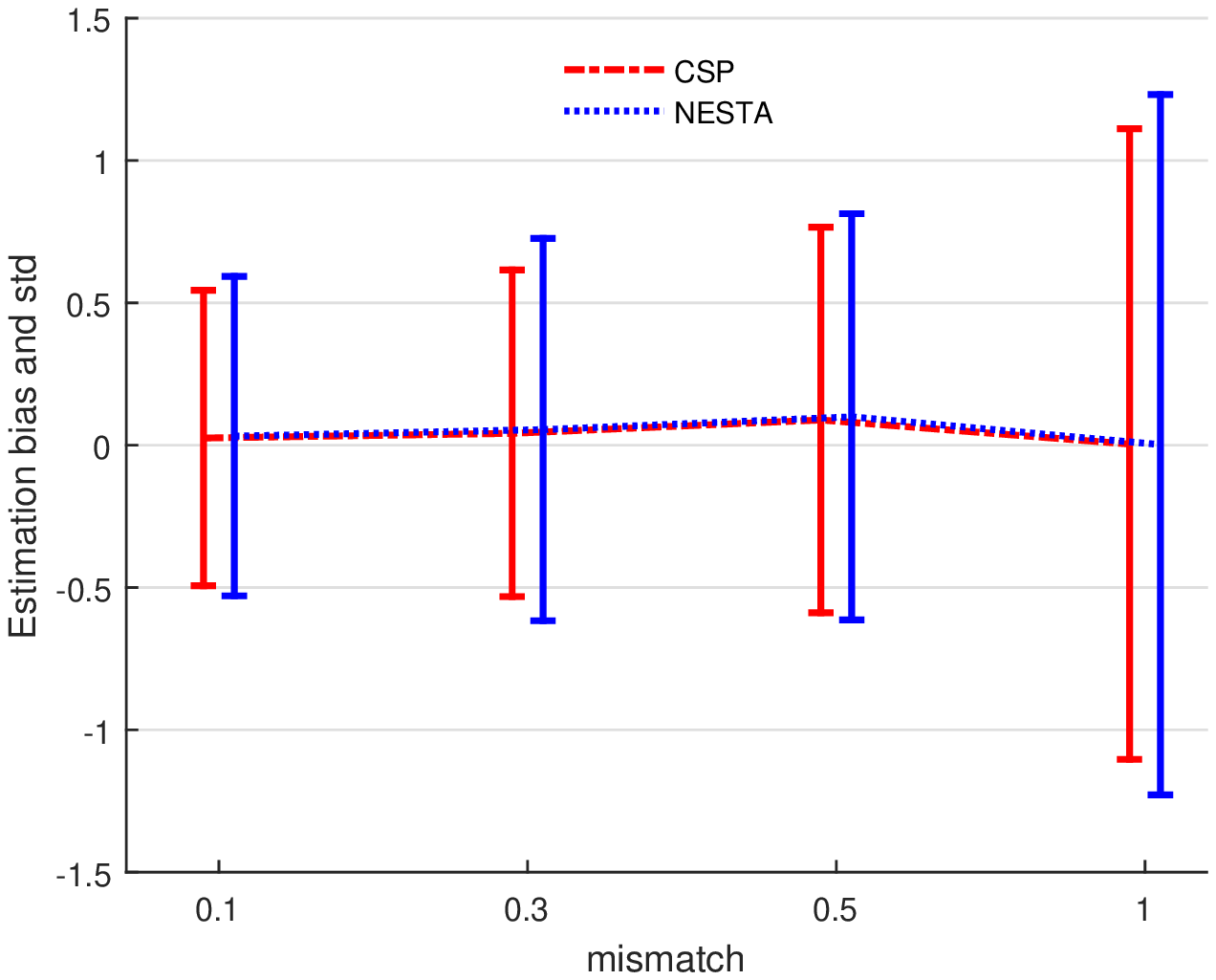}
\label{figure:offgrid_Nestab}} \\
\caption{Comparison of the proposed CSP approach with NESTA for different values of grid mismatch (a) ROC , (b) angle estimation bias and std.}
\label{figure:offgrid_Nesta}
\end{figure}

\subsection{Multi-Target Scenario}
Here, we evaluate the performance of the proposed algorithm in comparison with NESTA and OMP, considering multi-target scenarios. For the following simulations, the angular cells of the targets are randomly selected with a uniform distribution.
To provide a fair comparison, we assume that the number of targets $Q$ is known for all methods. Figure \ref{fig:multi-target-sample} depicts a sample multi-target scenario with $Q=5$. In this figure, different amplitudes are used to improve the display. Neither the proposed CSP approach nor the NESTA and OMP algorithms find the angles exactly. Still, the proposed method provides a better estimation.

\begin{figure}
	\centering
	\psfrag{T}{\scriptsize{$\theta$}}
	\psfrag{Original}{\scriptsize{Original}}
	\psfrag{CSP}{\scriptsize{CSP}}
	\psfrag{NESTA}{\scriptsize{NESTA}}
	\psfrag{OMP}{\scriptsize{OMP}}
	\includegraphics[width=.44\textwidth]{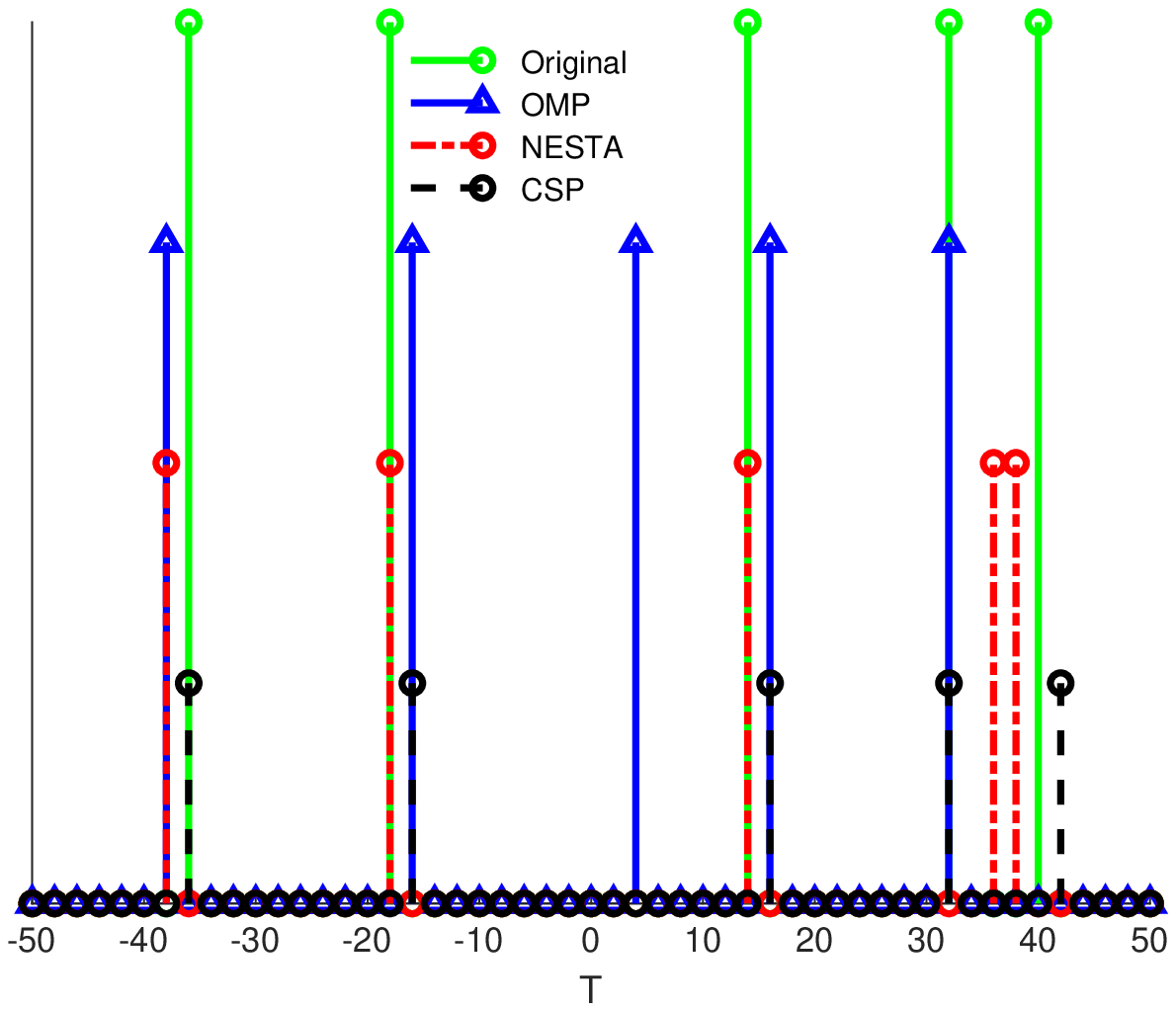}
	\caption{An example of a multi-target scenario. Comparison of target estimation accuracy of CSP, NESTA, and OMP algorithms.}
	\label{fig:multi-target-sample}
\end{figure}

Figure \ref{fig:multi-target-error} presents a comparison with NESTA and OMP, where estimation std versus number of targets $Q$ is plotted. As intuitively expected, increasing the number of targets, increases the estimation std. This issue is the result of correlation among columns of the dictionary matrix $\bTheta$ which leads to an estimation error in multi-target scenarios. In other words, a linear combination of a subset of columns of $\bTheta$ could be close to a linear combination of another subset of columns of $\bTheta$. This issue especially affects the performance of recovery-based algorithms. For instance, the objective of NESTA is to minimize the norm of the residual. With correlated columns of the dictionary matrix as well as noise and clutter contaminated measurements, this leads to a faulty reconstruction of the targets since the focus is on minimizing the difference between the synthesized data vector and the measured data vector. In contrast, for the proposed CSP based algorithm, at each step, the most likely column is selected, which leads to a smaller estimation error. Another reason for the better performance of the proposed method is the implicit whitening procedure in the procedure. Although employing the Capon beamformer, the clutter and noise power is reduced, yet the small residual effects the performance where the CSP method compensates it with an implicit whitening filter.
 As shown in Figure \ref{fig:multi-target-error}, increasing the SNR, reduces the estimation error. Also, in all scenarios, the proposed approach outperforms both NESTA and OMP.
\begin{figure}
	\centering
	\psfrag{Estimation error}{\scriptsize{Estimation std}}
	\psfrag{Q}{\scriptsize{$Q$}}
	\psfrag{CSP-SNR=10}{\scriptsize{CSP ~~~ - SNR = 0dB}}
	\psfrag{CSP-SNR=20}{\scriptsize{CSP ~~~ - SNR = 10dB}}
	\psfrag{OMP-SNR=10}{\scriptsize{OMP ~~~ - SNR = 0dB}}
	\psfrag{OMP-SNR=20}{\scriptsize{OMP ~~~ - SNR = 10dB}}
	\psfrag{NESTA-SNR=10}{\scriptsize{NESTA - SNR = 0dB}}
	\psfrag{NESTA-SNR=20}{\scriptsize{NESTA - SNR = 10dB}}
	\includegraphics[width=.44\textwidth]{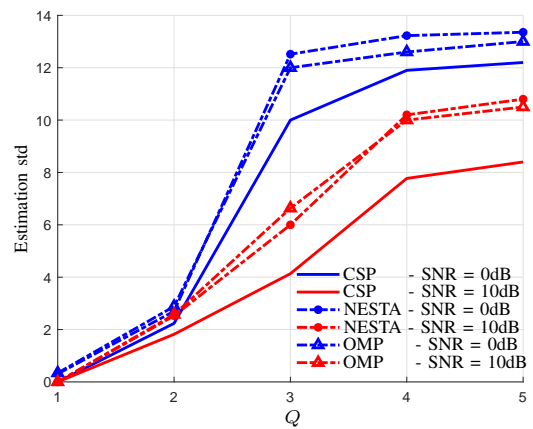}
	\caption{Comparison of the proposed CSP approach with NESTA and OMP based on angle estimation std versus number of targets for different values of SNR.}
	\label{fig:multi-target-error}
\end{figure}

For the next simulation, we aim to observe how the proposed algorithm perform when two targets are in close angle cells and also investigate the effect of the compression ratio on its performance. Therefore, we consider a scenario with two targets and changing their angle difference, i.e., $\Delta \theta = \theta_1-\theta_2$, we calculate the percentage of correct estimation of both angles denoted as $\rm{P}_{\rm{CE}}$. In this simulation, we assume $\rm{SNR} = 20dB$ and $R=20$. The rest of parameters are similar to Table \ref{tab:table1}. Figure \ref{fig:resolvability} presents the percentage of correct angle estimation versus the two targets angular difference for different values of the compression ratio. As expected, the estimation accuracy drops when the targets are closer. In addition, increasing the compression ratio, the resolvability of targets are reduced. {This reduction in resolvability is due to the loss in SNR, as SNR is directly proportional to the number of samples. In order to prevent such an SNR loss, techniques such as compressive data acquisition directly at reception are proved to be useful \cite{8450796}.}
\begin{figure}
	\centering
	\psfrag{delta}{\scriptsize{$\Delta \theta$}}
	\psfrag{pce}{\scriptsize{$\rm{P}_{\rm{CE}}$}}
	\psfrag{CR=1}{\scriptsize{$\rm{CR}_1 = 1$}}
	\psfrag{CR=2}{\scriptsize{$\rm{CR}_1 = 2$}}
	\psfrag{CR=4}{\scriptsize{$\rm{CR}_1 = 4$}}
	\psfrag{CR=8}{\scriptsize{$\rm{CR}_1 = 8$}}
	\psfrag{CR=16}{\scriptsize{$\rm{CR}_1 = 16$}}
	\includegraphics[width=.44\textwidth]{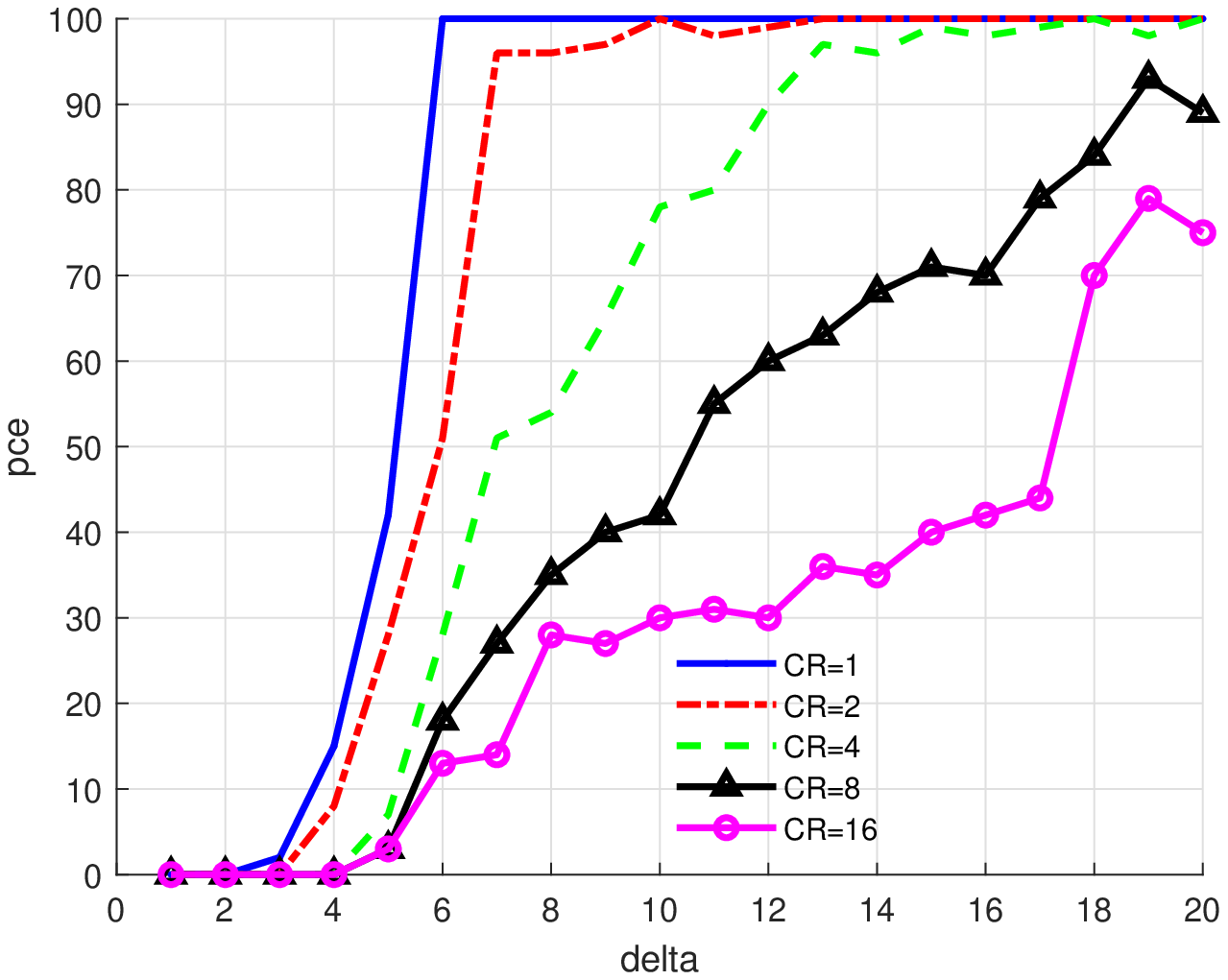}
	\caption{The percentage of correct angle estimation versus the targets angle difference for different values of compression ratio.}
	\label{fig:resolvability}
\end{figure}

\subsection{CSP MIMO versus CS MIMO}
Here, we aim to compare the proposed CSP MIMO radar versus a conventional CS MIMO radar in terms of the number of required samples conditioned on achieving the same performance. Considering the same first compression matrix and capon beamformer for both approaches, the second compression matrix is only employed for the CSP MIMO radar. Moreover, the second compression ratio $\rm{CR}_2$ is determined such that both methods achieve the same performance, i.e., estimation accuracy. For the following simulation, angle estimation accuracy is considered as the performance metric. 
As depicted in Figure \ref{fig:CSPvsCS}, $\rm{CR}_2$ starting from around 2 for $\rm{CR}_1=1$, we can increase $\rm{CR}_2$ up to 8 for $\rm{CR}_1=16$.

\begin{figure}
	\centering
	\psfrag{CR1}{\scriptsize{$\rm{CR}_1$}}
	\psfrag{CR2}{\scriptsize{$\rm{CR}_2$}}
	\includegraphics[width=.44\textwidth]{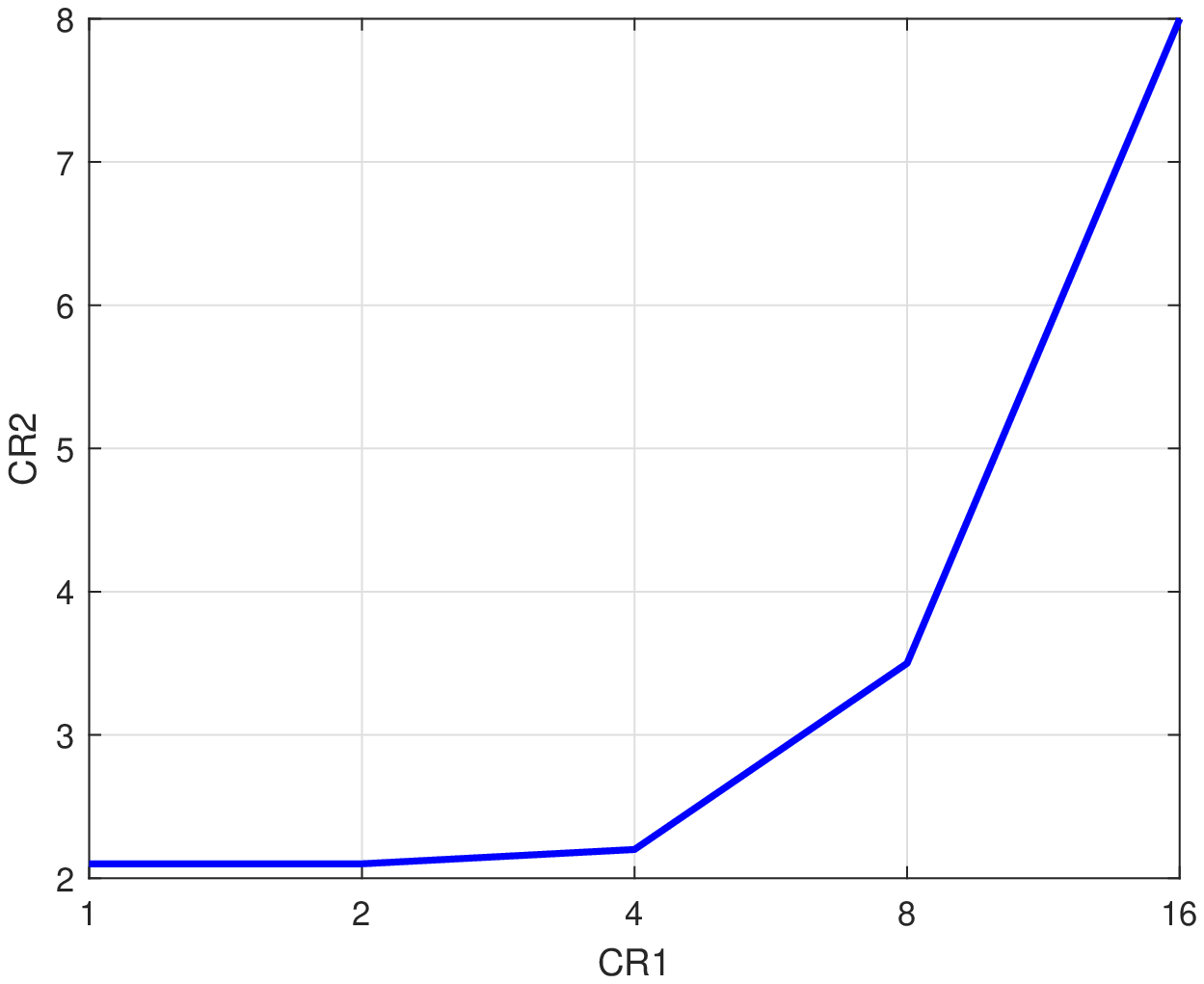}
	\caption{The second compression ratio $\rm{CR}_2$ versus the first compression ratio $\rm{CR}_1$ in order to achieve the same performance for both CSP MIMO and CS MIMO radar.}
	\label{fig:CSPvsCS}
\end{figure}

\section{Conclusion}
MIMO radar has been receiving a lot of attention for automotive radar applications. A high data rate and computational complexity are the main drawbacks of MIMO radar. The proposed method consists of performing temporal and spatial CS, applying the Capon beamformer to reduce the clutter, applying a second compression, and then, formulation and solving a target detection problem on each grid of the angle space. The proposed method achieves significant sample and computational complexity, and is particularly suited in applications that require low latency, such as automotive radar.
Through simulations, we have illustrated that performing the signal processing in the compressed domain not only reduces sample complexity, but also improves the detection probability and angle estimation accuracy, especially in multi-target scenarios. Additionally, we have provided a mathematical analysis for the detector's ROC that was well aligned with the simulation results.
{As future work we will implement the proposed algorithm over a test-bed using real-world data sets in order to obtain a more realistic evaluation.}
%As future work, we plan to expand the formulation of this paper to consider a widely separated MIMO radar configuration.

%\appendices
\appendix
%\section{}

\subsection{Proof of Theorem \ref{theorem:f(H1)}}
\label{proof 1}
Using the PDF of $\bz$ for hypothesis $\mathcal{H}_1$, conditioned on $\alpha$ and $t$ given in \eqref{target1}, we prove the PDF of $\bz$ under $\mathcal{H}_1$ is
\begin{equation}
\begin{aligned}
f(\bz|\mathcal{H}_1,t) &=  \frac{1}{\pi^L|\bA|}\frac{1}{\sigma_{\alpha}^2 d_t+1}\exp{\left(-\bz^H\bA^{-1}\bz\right)} \\&\exp\left(\frac{|e_t|^2\sigma_{\alpha}^2}{\sigma_{\alpha}^2 d_t + 1}\right),
\end{aligned}
\label{equ:H_1app}
\end{equation}
where $d_t$ and $e_t$ are defined in \eqref{equ:definitions}.
\begin{proof}
	\begin{equation}
	\begin{aligned}
	f(\bz|\mathcal{H}_1,t)
	&= \int f(\bz|\mathcal{H}_1,\alpha,t)f(\alpha) d\alpha \\
	&= \frac{1}{\pi^{M_2}|\bA|}\frac{1}{\pi \sigma_{\alpha}^2}\exp{\left(-\bz^H\bA^{-1}\bz\right)} \\&\int \exp\left(-\btheta_t^H \bPhi_{(2)}^T \bA^{-1} \bPhi_{(2)} \btheta_t |\alpha|^2\right) \\&\exp\left(\btheta_t^H \bPhi_{(2)}^T \bA^{-1} \bz \alpha^*\right) \\
	&\exp\left(\bz^H \bA^{-1} \bPhi_{(2)} \btheta_t \alpha\right) \exp\left(-\frac{|\alpha|^2}{\sigma_{\alpha}^2}\right) d\alpha.
	\end{aligned}
	\label{target2}
	\end{equation}
	
Note that $d_t = d_t^*$ by definition. Using the definitions in \eqref{equ:definitions} for $d_t$ and $e_t$, (\ref{target2}) can be reformulated as
\begin{equation}
\begin{aligned}
f(\bz|\mathcal{H}_1,t) &= \frac{1}{\pi^{(M_2)}|\bA|}\frac{1}{\pi \sigma_{\alpha}^2}\exp{\left(-\bz^H\bA^{-1}\bz\right)} \\&\int \exp\left(-(d_t+\frac{1}{\sigma_{\alpha}^2})|\alpha|^2\right) \\&\exp\left(e_t \alpha^* + e_t^* \alpha\right) d\alpha \\
&= \frac{1}{\pi^{M_2}|\bA|}\frac{1}{\pi \sigma_{\alpha}^2}\exp{\left(-\bz^H\bA^{-1}\bz\right)} \\&\int \exp\left(-g_t(|\alpha|^2-\frac{e_t\alpha^*}{g_t}-\frac{e_t^*\alpha}{g_t})\right) d\alpha,
\end{aligned}
\end{equation}
where
\begin{equation}
g_t = d_t + \frac{1}{\sigma_{\alpha}^2}.
\end{equation}
Thus
\begin{equation}
\begin{aligned}
f(\bz|\mathcal{H}_1,t) &=  \frac{1}{\pi^{M_2}|\bA|}\frac{1}{\pi \sigma_{\alpha}^2}\exp{\left(-\bz^H\bA^{-1}\bz\right)} \\&\int \exp\left(-g_t(|\alpha-\frac{e_t}{g_t}|^2-\frac{|e_t|^2}{g_t^2})\right) d\alpha \\
&=  \frac{1}{\pi^{M_2}|\bA|}\frac{1}{\sigma_{\alpha}^2 g_t}\exp{\left(-\bz^H\bA^{-1}\bz\right)} \\&\exp\left(\frac{|e_t|^2}{g_t}\right),
\end{aligned}
\end{equation}
which results in
\begin{equation}
\begin{aligned}
f(\bz|\mathcal{H}_1,t) &=  \frac{1}{\pi^{M_2}|\bA|}\frac{1}{\sigma_{\alpha}^2 d_t+1}\exp{\left(-\bz^H\bA^{-1}\bz\right)} \\&\exp\left(\frac{\sigma_{\alpha}^2 |e_t|^2}{\sigma_{\alpha}^2 d_t + 1}\right).
\end{aligned}
\end{equation}
\end{proof}

\subsection{Proof of (\ref{equ:GLRT})}
\label{proof 2}
The LRT is derived as
\begin{equation}
\begin{aligned}
{\rm{L}}(\bz|t) &= \frac{f(\bz|\mathcal{H}_1,t)}{f(\bz|\mathcal{H}_0)} \\&= \frac{1}{d_t\sigma_{\alpha}^2+1} \exp\left(\frac{|e_t|^2\sigma_{\alpha}^2}{d_t\sigma_{\alpha}^2+1}\right) \gtrless \frac{p_0}{1-p_0},
\end{aligned}
\end{equation}
where $p_0$ is the a priori probability of the $\mathcal{H}_0$ hypothesis. It is clear that $d_t$ is not dependent on the measurement vector and we can simplify the LRT through the following steps:
\begin{equation}
{\rm L_1}(\bz|t) = \exp\left(\frac{|e_t|^2\sigma_{\alpha}^2}{d_t\sigma_{\alpha}^2+1}\right) \gtrless \frac{p_0}{1-p_0} (d_t\sigma_{\alpha}^2+1),
\end{equation}
and
\begin{equation}
{\rm{L_2}}(\bz|t) = \frac{|e_t|^2\sigma_{\alpha}^2}{d_t\sigma_{\alpha}^2+1} \gtrless \ln\left(\frac{p_0}{1-p_0} (d_t\sigma_{\alpha}^2+1)\right).
\end{equation}
As a result
\begin{equation}
{\rm{L_3}}(\bz|t) = |e_t|^2 \gtrless \frac{d_t\sigma_{\alpha}^2+1}{\sigma_{\alpha}^2}\ln\left(\frac{p_0}{1-p_0} (d_t\sigma_{\alpha}^2+1)\right),
\end{equation}
which leads to
\begin{equation}
\begin{aligned}
{\rm{L_4}}(\bz|t) &= |e_t| = |\btheta_t^H \bPhi_{(2)}^T \bA^{-1} \bz| \gtrless \eta,
\end{aligned}
\label{equ:L4}
\end{equation}
where we define $\eta$ as the detection threshold, which is determined based on the desired false alarm probability $P_{fa}$.

If ${\rm{L_4}}(\bz|t)$ is maximized over $t$, the GLRT will be ${\rm GLRT}(\bz) = \underset{t\in\{1,\ldots,L \}}{\max}~{\rm{L_4}}(\bz|t)$ in which $\hat{t} = \underset{t\in\{1,\ldots,L \}}{\arg \max}~{\rm{L_4}}(\bz|t)$. Hence, the GLRT can be written as

\begin{equation}
\begin{aligned}
{\rm{GLRT}}(\bz) = {\rm{L_4}}(\bz|\hat{t})  =  |e_{\hat{t}}| \gtrless \eta,
\end{aligned}
\end{equation}

\subsection{Proof of (\ref{equ:hH_1})}
\label{proof 3}
The PDF of $e_{\hat{t}}$ conditioned on $\mathcal{H}_1$ can be computed as

\begin{equation}
\begin{aligned}
f(e_{\hat{t}}|\mathcal{H}_1) &= \int f(e_{\hat{t}}|\mathcal{H}_1,\alpha)f(\alpha)d\alpha \\&= \frac{1}{\pi^2 d_{\hat{t}} \sigma_{\alpha}^2} \\&\int \exp\left(-\frac{|e_{\hat{t}}-\alpha d_{\hat{t}}|^2}{d_{\hat{t}}}\right) \exp\left(-\frac{|\alpha|^2}{\sigma_{\alpha}^2}\right) d\alpha\\
 &= \frac{1}{\pi^2 d_{\hat{t}} \sigma_{\alpha}^2}\exp\left(-\frac{|e_{\hat{t}}|^2}{d_{\hat{t}}}\right) \\&\int \exp\left(-(d_{\hat{t}}+\frac{1}{\sigma_{\alpha}^2})|\alpha|^2\right) \\&\exp\left(e_{\hat{t}}^*\alpha\right) \exp\left(e_{\hat{t}}\alpha^*\right) d\alpha \\
&= \frac{1}{\pi^2 d_{\hat{t}} \sigma_{\alpha}^2}\exp\left(-\frac{|e_{\hat{t}}|^2}{d_{\hat{t}}}\right)\\&\int \exp\left(-g_{\hat{t}}(|\alpha|^2-\frac{e_{\hat{t}}^*\alpha}{g_{\hat{t}}}-\frac{e_{\hat{t}}\alpha^*}
{g_{\hat{t}}})\right) d\alpha \\
&= \frac{1}{\pi^2 d_{\hat{t}} \sigma_{\alpha}^2}\exp\left(-\frac{|e_{\hat{t}}|^2}{d_{\hat{t}}}\right)\\&\int \exp\left(-g_{\hat{t}}(|\alpha-\frac{e_{\hat{t}}}{g_{\hat{t}}}|^2-\frac{|e_{\hat{t}}|^2}{g_{\hat{t}}^2})\right) d\alpha \\
&= \frac{1}{\pi^2 d_{\hat{t}} \sigma_{\alpha}^2}\exp\left(-\frac{|e_{\hat{t}}|^2}{d_{\hat{t}}}\right) \exp\left(\frac{|e_{\hat{t}}|^2}{g_{\hat{t}}}\right) \\&\int \exp\left(-g_{\hat{t}}(|\alpha-\frac{e_{\hat{t}}}{g_{\hat{t}}}|^2)\right)  d\alpha \\
&= \frac{1}{\pi^2 d_{\hat{t}} \sigma_{\alpha}^2}\exp\left(-\frac{|e_{\hat{t}}|^2}{d_{\hat{t}}g_{\hat{t}}\sigma_{\alpha}^2}\right) \pi g_{\hat{t}} \\
&= \frac{1}{\pi d_{\hat{t}} (d_{\hat{t}} \sigma_{\alpha}^2+1)}\exp\left(-\frac{|e_{\hat{t}}|^2}{d_{\hat{t}}(d_{\hat{t}}\sigma_{\alpha}^2+1)}\right).
\end{aligned}
\end{equation}

\subsection{Proof of \eqref{derxmidmain}}
\label{proof 4}
Assuming $t=\hat{t}$, we will have
\begin{equation}
\begin{aligned}
x &= |\alpha\btheta^H_t \bPhi_{(2)}^T \bA^{-1} \bPhi_{(2)}\btheta_t + \btheta^H_t \bPhi_{(2)}^T \bA^{-1}\bPhi_{(2)}\bnu| \\&= |\alpha d_t + \btheta^H_t \bPhi_{(2)}^T \bA^{-1}\bPhi_{(2)}\bnu|.
\end{aligned}
\end{equation}
Subsequently, $x^2$ can be expressed as
\begin{equation}
\begin{aligned}
x^2 &= |\alpha|^2 d_t^2 + \alpha^*d_t\btheta^H_t \bPhi_{(2)}^T \bA^{-1}\bPhi_{(2)}\bnu \\&+ \alpha d_t \bnu^H \bPhi_{(2)}^T \bA^{-1}\bPhi_{(2)} \btheta_t \\&+ 
\btheta^H_t \bPhi_{(2)}^T \bA^{-1}\bPhi_{(2)}\bnu \bnu^H \bPhi_{(2)}^T \bA^{-1}\bPhi_{(2)} \btheta_t.
\end{aligned}
\end{equation}
Since the Capon beamformer is applied to obtain $\by$, we assume that $||\alpha\btheta_t||_2 \gg ||\bnu||$. As a consequence, $x$ can be approximated as follows
\begin{equation}
\begin{aligned}
x &= |\alpha| d_t\Bigg(1+\frac{\btheta^H_t \bPhi_{(2)}^T \bA^{-1}\bPhi_{(2)}\bnu}{\alpha d_t} \\&+ \frac{\bnu^H \bPhi_{(2)}^T \bA^{-1}\bPhi_{(2)} \btheta_t}{\alpha^*d_t} \\&+ 
\frac{\btheta^H_t \bPhi_{(2)}^T \bA^{-1}\bPhi_{(2)}\bnu \bnu^H \bPhi_{(2)}^T \bA^{-1}\bPhi_{(2)} \btheta_t}{|\alpha|^2d_t^2}\Bigg)^{0.5}\\ 
&\approxeq |\alpha| d_t\Bigg(1+\frac{\btheta^H_t \bPhi_{(2)}^T \bA^{-1}\bPhi_{(2)}\bnu}{\alpha d_t} \\&+ \frac{\bnu^H \bPhi_{(2)}^T \bA^{-1}\bPhi_{(2)} \btheta_t}{\alpha^*d_t}\Bigg)^{0.5}\\
&= |\alpha| d_t\Bigg(1+2\Re\bigg\{\frac{\btheta^H_t \bPhi_{(2)}^T \bA^{-1}\bPhi_{(2)}\bnu}{\alpha d_t}\bigg\}\Bigg)^{0.5} \\
&\approxeq |\alpha| d_t\Bigg(1+\Re\bigg\{\frac{\btheta^H_t \bPhi_{(2)}^T \bA^{-1}\bPhi_{(2)}\bnu}{\alpha d_t}\bigg\}\Bigg)\\
&= |\alpha| d_t+\Re\bigg\{\frac{|\alpha|}{\alpha}\btheta^H_t \bPhi_{(2)}^T \bA^{-1}\bPhi_{(2)}\bnu\bigg\}.
\end{aligned}
\label{derxmid}
\end{equation}

% you can choose not to have a title for an appendix
% if you want by leaving the argument blank
%\section{}
%Appendix two text goes here.

% use section* for acknowledgement
%\section*{Acknowledgment}

%The authors would like to thank...
%\section*{Acknowledgment}
%The authors would like to thank Dr. Augusto Aubry and Prof. Antonio De Maio from Universit\`a degli Studi di Napoli ``Federico II'', for their helpful and constructive comments that greatly contributed to improve the quality of this paper.

\ifCLASSOPTIONcaptionsoff
\newpage
\fi
\bibliographystyle{ieeetr}
\bibliography{ref}

\begin{thebibliography}{10}

\bibitem{haimovich2008mimo}
A.~M. Haimovich, R.~S. Blum, and L.~J. Cimini, ``{MIMO} radar with widely
  separated antennas,'' {\em IEEE Signal Processing Magazine}, vol.~25, no.~1,
  pp.~116--129, 2008.

\bibitem{Radieee2}
M.~Radmard, M.~M. Chitgarha, M.~N. Majd, and M.~M. Nayebi, ``Antenna placement
  and power allocation optimization in {MIMO} detection,'' {\em IEEE
  Transactions on Aerospace and Electronic Systems}, vol.~50, pp.~1468--1478,
  April 2014.

\bibitem{tohidi2017compressive}
E.~Tohidi, M.~Radmard, M.~N. Majd, H.~Behroozi, and M.~M. Nayebi, ``Compressive
  sensing {MTI} processing in distributed {MIMO} radars,'' {\em IET Signal
  Processing}, vol.~12, no.~3, pp.~327--334, 2017.

\bibitem{li2007mimo}
J.~Li and P.~Stoica, ``{MIMO} radar with colocated antennas,'' {\em IEEE Signal
  Processing Magazine}, vol.~24, pp.~106--114, Sept 2007.

\bibitem{karbasi2015design}
S.~M. Karbasi, M.~Radmard, M.~M. Nayebi, and M.~H. Bastani, ``Design of
  multiple-input multiple-output transmit waveform and receive filter for
  extended target detection,'' {\em IET Radar, Sonar \& Navigation}, vol.~9,
  no.~9, pp.~1345--1353, 2015.

\bibitem{8335404}
E.~Tohidi, H.~Behroozi, and G.~Leus, ``Antenna and pulse selection for
  colocated {MIMO} radar,'' in {\em 51st Asilomar Conference on Signals,
  Systems, and Computers}, pp.~563--567, Oct 2017.

\bibitem{8706630}
K.~{Alhujaili}, V.~{Monga}, and M.~{Rangaswamy}, ``Transmit {MIMO} radar
  beampattern design via optimization on the complex circle manifold,'' {\em
  IEEE Transactions on Signal Processing}, vol.~67, pp.~3561--3575, July 2019.

\bibitem{8537943}
E.~Tohidi, M.~Coutino, S.~P. Chepuri, H.~Behroozi, M.~M. Nayebi, and G.~Leus,
  ``Sparse antenna and pulse placement for colocated {MIMO} radar,'' {\em IEEE
  Transactions on Signal Processing}, vol.~67, pp.~579--593, Feb 2019.

\bibitem{rossi2014spatial}
M.~Rossi, A.~M. Haimovich, and Y.~C. Eldar, ``Spatial compressive sensing for
  {MIMO} radar,'' {\em IEEE Transactions on Signal Processing}, vol.~62,
  pp.~419--430, Jan 2014.

\bibitem{ender2010compressive}
J.~Ender, ``On compressive sensing applied to radar,'' {\em Signal Processing},
  vol.~90, no.~5, pp.~1402--1414, 2010.

\bibitem{tohidicosera}
E.~Tohidi, M.~Radmard, S.~M. Karbasi, H.~Behroozi, and M.~M. Nayebi,
  ``Compressive sensing in {MTI} processing,'' in {\em 3rd International
  Workshop on Compressed Sensing Theory and its Applications to Radar, Sonar
  and Remote Sensing (CoSeRa)}, pp.~189--193, June 2015.

\bibitem{biondi2015compressed}
F.~Biondi, ``Compressed sensing radar-new concepts of incoherent continuous
  wave transmissions,'' in {\em 2015 3rd International Workshop on Compressed
  Sensing Theory and its Applications to Radar, Sonar and Remote Sensing
  (CoSeRa)}, pp.~204--208, IEEE, 2015.

\bibitem{7745736}
C.~{Larsson}, ``Compressive sensing methods for radar cross section {ISAR}
  measurements,'' in {\em 2016 4th International Workshop on Compressed Sensing
  Theory and its Applications to Radar, Sonar and Remote Sensing (CoSeRa)},
  pp.~237--241, Sep. 2016.

\bibitem{8514046}
I.~{Taghavi}, M.~F. {Sabahi}, and F.~{Parvaresh}, ``High resolution compressed
  sensing radar using difference set codes,'' {\em IEEE Transactions on Signal
  Processing}, vol.~67, pp.~136--148, Jan 2019.

\bibitem{8639010}
W.~{Feng}, J.~{Friedt}, G.~{Cherniak}, and M.~{Sato}, ``Batch compressive
  sensing for passive radar range-doppler map generation,'' {\em IEEE
  Transactions on Aerospace and Electronic Systems}, vol.~55, pp.~3090--3102,
  Dec 2019.

\bibitem{baraniuk2007compressive}
R.~G. Baraniuk, ``Compressive sensing [lecture notes],'' {\em IEEE Signal
  Processing Magazine}, vol.~24, pp.~118--121, July 2007.

\bibitem{emmanuel2004robust}
E.~J. {Candes}, J.~{Romberg}, and T.~{Tao}, ``Robust uncertainty principles:
  exact signal reconstruction from highly incomplete frequency information,''
  {\em IEEE Transactions on Information Theory}, vol.~52, no.~2, pp.~489--509,
  2006.

\bibitem{yu2010mimo}
Y.~Yu, A.~P. Petropulu, and H.~V. Poor, ``{MIMO radar using compressive
  sampling},'' {\em IEEE Journal of Selected Topics in Signal Processing},
  vol.~4, no.~1, pp.~146--163, 2010.

\bibitem{7485215}
I.~{Bilik}, O.~{Bialer}, S.~{Villeval}, H.~{Sharifi}, K.~{Kona}, M.~{Pan},
  D.~{Persechini}, M.~{Musni}, and K.~{Geary}, ``Automotive {MIMO} radar for
  urban environments,'' in {\em 2016 IEEE Radar Conference (RadarConf)},
  pp.~1--6, May 2016.

\bibitem{alland2009radar}
S.~W. Alland and J.~F. Searcy, ``Radar system and method of digital
  beamforming,'' Dec.~29 2009.
\newblock US Patent 7,639,171.

\bibitem{wintermantel2014radar}
M.~Wintermantel, ``Radar system with improved angle formation,'' Mar.~4 2014.
\newblock US Patent 8,665,137.

\bibitem{schoor2016method}
M.~Schoor, G.~Kuehnle, K.~Rambach, and B.~Loesch, ``Method for operating a
  {MIMO} radar,'' Sept.~20 2016.
\newblock US Patent 9,448,302.

\bibitem{8361480}
D.~Cohen, D.~Cohen, Y.~C. Eldar, and A.~M. Haimovich, ``{SUMMeR}: {Sub-Nyquist}
  {MIMO} radar,'' {\em IEEE Transactions on Signal Processing}, pp.~1--1, 2018.

\bibitem{7467561}
P.~Chen, C.~Qi, L.~Wu, and X.~Wang, ``Estimation of extended targets based on
  compressed sensing in cognitive radar system,'' {\em IEEE Transactions on
  Vehicular Technology}, vol.~66, pp.~941--951, Feb 2017.

\bibitem{gogineni2011target}
S.~Gogineni and A.~Nehorai, ``Target estimation using sparse modeling for
  distributed {MIMO} radar,'' {\em IEEE Transactions on Signal Processing},
  vol.~59, pp.~5315--5325, Nov 2011.

\bibitem{7376237}
B.~Li and A.~P. Petropulu, ``Distributed {MIMO} radar based on sparse sensing:
  Analysis and efficient implementation,'' {\em IEEE Transactions on Aerospace
  and Electronic Systems}, vol.~51, pp.~3055--3070, Oct 2015.

\bibitem{herman2009high}
M.~A. Herman and T.~Strohmer, ``High-resolution radar via compressed sensing,''
  {\em IEEE Transactions on Signal Processing}, vol.~57, pp.~2275--2284, June
  2009.

\bibitem{yu2014power}
Y.~Yu, S.~Sun, R.~N. Madan, and A.~Petropulu, ``Power allocation and waveform
  design for the compressive sensing based {MIMO} radar,'' {\em IEEE
  Transactions on Aerospace and Electronic Systems}, vol.~50, pp.~898--909,
  April 2014.

\bibitem{lei2016compressed}
L.~Lei, J.~Huang, and Y.~Sun, ``Compressed sensing mimo radar waveform
  optimization without signal recovery,'' in {\em 2016 CIE International
  Conference on Radar (RADAR)}, pp.~1--4, IEEE, 2016.

\bibitem{8468214}
A.~{Ajorloo}, A.~{Amini}, and M.~H. {Bastani}, ``A compressive sensing-based
  colocated {MIMO} radar power allocation and waveform design,'' {\em IEEE
  Sensors Journal}, vol.~18, pp.~9420--9429, Nov 2018.

\bibitem{8395364}
S.~{Salari}, F.~{Chan}, Y.~{Chan}, I.~{Kim}, and R.~{Cormier}, ``Joint {DOA}
  and clutter covariance matrix estimation in compressive sensing {MIMO}
  radar,'' {\em IEEE Transactions on Aerospace and Electronic Systems},
  vol.~55, pp.~318--331, Feb 2019.

\bibitem{needell2009cosamp}
D.~Needell and J.~A. Tropp, ``{CoSaMP}: Iterative signal recovery from
  incomplete and inaccurate samples,'' {\em Applied and Computational Harmonic
  Analysis}, vol.~26, no.~3, pp.~301--321, 2009.

\bibitem{huang2012adaptive}
T.~Huang, Y.~Liu, and X.~Wang, ``Adaptive subspace pursuit and its application
  in motion compensation for step frequency radar,'' in {\em 1st International
  Workshop on Compressed Sensing applied to Radar (CoSeRa)}, 2012.

\bibitem{zhang2011sparse}
B.~Zhang, X.~Cheng, N.~Zhang, Y.~Cui, Y.~Li, and Q.~Liang, ``Sparse target
  counting and localization in sensor networks based on compressive sensing,''
  in {\em 2011 Proceedings IEEE INFOCOM}, pp.~2255--2263, April 2011.

\bibitem{davenport2010signal}
M.~A. Davenport, P.~T. Boufounos, M.~B. Wakin, and R.~G. Baraniuk, ``Signal
  processing with compressive measurements,'' {\em IEEE Journal of Selected
  Topics in Signal Processing}, vol.~4, pp.~445--460, April 2010.

\bibitem{hariri2017compressive}
A.~Hariri and M.~Babaie-Zadeh, ``Compressive detection of sparse signals in
  additive white gaussian noise without signal reconstruction,'' {\em Signal
  Processing}, vol.~131, pp.~376--385, 2017.

\bibitem{6797969}
S.~{Gishkori}, V.~{Lottici}, and G.~{Leus}, ``Compressive sampling-based
  multiple symbol differential detection for {UWB} communications,'' {\em IEEE
  Transactions on Wireless Communications}, vol.~13, pp.~3778--3790, July 2014.

\bibitem{hariri2015joint}
A.~Hariri and M.~Babaie-Zadeh, ``Joint compressive single target detection and
  parameter estimation in radar without signal reconstruction,'' {\em IET
  Radar, Sonar \& Navigation}, vol.~9, no.~8, pp.~948--955, 2015.

\bibitem{wicks2012compressed}
M.~C. Wicks, H.~Kung, and H.-C. Chen, ``Compressed statistical testing and
  application to radar,'' 2012.

\bibitem{lim2012automatic}
C.~W. Lim and M.~B. Wakin, ``Automatic modulation recognition for spectrum
  sensing using nonuniform compressive samples,'' in {\em 2012 IEEE
  International Conference on Communications (ICC)}, pp.~3505--3510, June 2012.

\bibitem{lim2012chocs}
C.~W. Lim and M.~B. Wakin, ``{CHOCS}: A framework for estimating compressive
  higher order cyclostationary statistics,'' in {\em SPIE Defense, Security,
  and Sensing}, pp.~83650M--83650M, International Society for Optics and
  Photonics, 2012.

\bibitem{5604239}
Y.~{Wang}, A.~{Pandharipande}, and G.~{Leus}, ``Compressive sampling based
  {MVDR} spectrum sensing,'' in {\em 2010 2nd International Workshop on
  Cognitive Information Processing}, pp.~333--337, June 2010.

\bibitem{skolnik2001introduction}
M.~Skolnik, {\em Introduction to Radar Systems}.
\newblock Electrical engineering series, McGraw-Hill, 2001.

\bibitem{4358826}
V.~Kovalenko, A.~G. Yarovoy, and L.~P. Ligthart, ``A novel clutter suppression
  algorithm for landmine detection with {GPR},'' {\em IEEE Transactions on
  Geoscience and Remote Sensing}, vol.~45, pp.~3740--3751, Oct 2007.

\bibitem{4768708}
X.~Meng, T.~Wang, J.~Wu, and Z.~Bao, ``Short-range clutter suppression for
  airborne radar by utilizing prefiltering in elevation,'' {\em IEEE Geoscience
  and Remote Sensing Letters}, vol.~6, pp.~268--272, April 2009.

\bibitem{6589175}
T.~K. Sjögren, V.~T. Vu, M.~I. Pettersson, F.~Wang, D.~J.~G. Murdin,
  A.~Gustavsson, and L.~M.~H. Ulander, ``Suppression of clutter in multichannel
  {SAR} {GMTI},'' {\em IEEE Transactions on Geoscience and Remote Sensing},
  vol.~52, pp.~4005--4013, July 2014.

\bibitem{7470499}
B.~Tang and J.~Tang, ``Joint design of transmit waveforms and receive filters
  for {MIMO} radar space-time adaptive processing,'' {\em IEEE Transactions on
  Signal Processing}, vol.~64, pp.~4707--4722, Sept 2016.

\bibitem{yu2013capon}
Y.~Yu, S.~Sun, and A.~P. Petropulu, ``A {Capon} beamforming method for clutter
  suppression in colocated compressive sensing based {MIMO} radars,'' in {\em
  SPIE Defense, Security, and Sensing}, pp.~87170J--87170J, International
  Society for Optics and Photonics, 2013.

\bibitem{6650099}
D.~S. Kalogerias and A.~P. Petropulu, ``Matrix completion in colocated {MIMO}
  radar: Recoverability, bounds theoretical guarantees,'' {\em IEEE
  Transactions on Signal Processing}, vol.~62, pp.~309--321, Jan 2014.

\bibitem{7272834}
S.~Sun, W.~U. Bajwa, and A.~P. Petropulu, ``{MIMO-MC} radar: A {MIMO} radar
  approach based on matrix completion,'' {\em IEEE Transactions on Aerospace
  and Electronic Systems}, vol.~51, pp.~1839--1852, July 2015.

\bibitem{van2002optimum}
H.~L. Van~Trees, {\em Optimum array processing: Part IV of detection,
  estimation and modulation theory}, vol.~1.
\newblock Wiley Online Library, 2002.

\bibitem{4104231}
R.~{Nitzberg}, ``Clutter map cfar analysis,'' {\em IEEE Transactions on
  Aerospace and Electronic Systems}, vol.~AES-22, no.~4, pp.~419--421, 1986.

\bibitem{papoulis2002probability}
A.~Papoulis and S.~U. Pillai, {\em Probability, random variables, and
  stochastic processes}.
\newblock Tata McGraw-Hill Education, 2002.

\bibitem{poor2013introduction}
H.~V. Poor, {\em An introduction to signal detection and estimation}.
\newblock Springer Science \& Business Media, 2013.

\bibitem{4385788}
J.~A. {Tropp} and A.~C. {Gilbert}, ``Signal recovery from random measurements
  via orthogonal matching pursuit,'' {\em IEEE Transactions on Information
  Theory}, vol.~53, pp.~4655--4666, Dec 2007.

\bibitem{becker2011nesta}
S.~Becker, J.~Bobin, and E.~J. Cand{\`e}s, ``{NESTA}: a fast and accurate
  first-order method for sparse recovery,'' {\em SIAM Journal on Imaging
  Sciences}, vol.~4, no.~1, pp.~1--39, 2011.

\bibitem{5710590}
Y.~Chi, L.~L. Scharf, A.~Pezeshki, and A.~R. Calderbank, ``Sensitivity to basis
  mismatch in compressed sensing,'' {\em IEEE Transactions on Signal
  Processing}, vol.~59, pp.~2182--2195, May 2011.

\bibitem{6576276}
G.~Tang, B.~N. Bhaskar, P.~Shah, and B.~Recht, ``Compressed sensing off the
  grid,'' {\em IEEE Transactions on Information Theory}, vol.~59,
  pp.~7465--7490, Nov 2013.

\bibitem{6320676}
Z.~Yang, L.~Xie, and C.~Zhang, ``Off-grid direction of arrival estimation using
  sparse bayesian inference,'' {\em IEEE Transactions on Signal Processing},
  vol.~61, pp.~38--43, Jan 2013.

\bibitem{5706373}
H.~{Zhu}, G.~{Leus}, and G.~B. {Giannakis}, ``Sparsity-cognizant total
  least-squares for perturbed compressive sampling,'' {\em IEEE Transactions on
  Signal Processing}, vol.~59, pp.~2002--2016, May 2011.

\bibitem{8450796}
R.~{Pribić}, G.~{Leus}, and C.~{Tzotzadinis}, ``Signal-to-noise-ratio analysis
  of compressive data acquisition,'' in {\em 2018 IEEE Statistical Signal
  Processing Workshop (SSP)}, pp.~603--607, 2018.

\end{thebibliography}

\end{document}